\documentclass[runningheads]{llncs}
\usepackage{graphicx}
\usepackage{amsmath}
\usepackage{amsfonts}
\usepackage{stmaryrd}
\usepackage[T1]{fontenc}
\usepackage{xcolor}
\usepackage[numbers]{natbib}
\usepackage{algorithm}
\usepackage{algorithmic}
\usepackage{enumerate}
\usepackage{tabularray}
\usepackage{multirow}
\usepackage{hyperref}
\renewcommand{\bibname}{References}
\DeclareMathOperator*{\argmax}{arg\,max}

\newcommand{\lstar}{\ensuremath{\textrm{L}^*}}
\newcommand{\ouralgorithm}{\textsc{Quintic}}

\newcommand{\obstable}{\langle S, E, T; \mathcal{C}, \Gamma\rangle}

\newcommand{\orow}[1]{\textbf{\textit{row}}(#1)}
\newcommand{\orows}[1]{\textbf{\textit{rows}}(#1)}
\newcommand{\repr}[1]{\llbracket #1\rrbracket}

\begin{document}
\title{Learning Quantitative Automata Modulo Theories}
\author{Eric Hsiung, Swarat Chaudhuri, and Joydeep Biswas}
\authorrunning{E. Hsiung et al.}
\institute{University of Texas at Austin\\\email{\{ehsiung,swarat,joydeepb\}@cs.utexas.edu}}
\maketitle              % typeset the header of the contribution

\begin{abstract}
Quantitative automata are useful representations for numerous applications, including modeling probability distributions over sequences to Markov chains and reward machines. 
Actively learning such automata typically occurs using explicitly gathered input-output examples under adaptations of the \lstar algorithm. 
However, obtaining explicit input-output pairs can be expensive, and there exist scenarios, including preference-based learning or learning from rankings, where providing constraints is a less exerting and a more natural way to concisely describe desired properties. Consequently, we propose the problem of learning deterministic quantitative 
automata from sets of constraints over the valuations of input sequences. 
We present \ouralgorithm{}, an active learning algorithm, wherein the learner infers a valid automaton through deductive reasoning, by applying a theory to a set of currently available constraints and an assumed preference model and quantitative automaton class. 
\ouralgorithm{} performs a complete search over the space of automata, and is guaranteed to be minimal and correctly terminate. 
Our evaluations utilize theory of rationals in order to learn summation, discounted summation, product, and classification quantitative automata, and indicate \ouralgorithm{} is effective at learning these types of automata.
\end{abstract}

\section{Introduction}

Quantitative automata~\cite{qvsw,chalupa2024quak} are an extremely expressive family of automata that use a valuation function to map sequences to a scalar value. Prominently used valuation functions used include average, discounted summation~\cite{discounted}, product, and others~\cite{qvsw}. They naturally capture a subset of weighted automata \cite{qvsw}, and they can be used to model deterministic finite automata by choosing an appropriate valuation function. They have been used for modeling probability distributions over sequences \cite{Weiss2019}, Markov chains~\cite{markovchain}, and discounted returns from reward machines \cite{icarte2018}.

The majority of methods for learning quantitative automata rely on actively obtained input-output examples. Explicit collection of input-output examples can be expensive, especially for rare behaviors, or if the input-output sequences are dangerous to a system or user. In contrast, collecting preference and ranking information about sequences has become increasingly prevalent in a variety of fields, including in domains such as robotics~\cite{narcomey2024learninghumanpreferencesrobot,MACCARINI20227,maroto2022}, recommendation systems~\cite{xueprefrec2023kdd, de2009preference}, reinforcement learning~\cite{wirth2016model,wirth2017survey,biyik2024batch,Sadigh2017ActivePL,christiano2017deep}, and tuning language models~\cite{stiennon2020learning}. Such comparison-based data is being embraced due to its ease of collection compared to input-output examples. When translated into constraints, preference and ranking information can be used to concisely describe desired system properties. Despite the move towards learning from such data, few methods have explored learning automata from such comparison information.

Consequently, we consider the problem of learning deterministic quantitative automata from comparisons of the valuations of sequences, and present an active learning solution. In our algorithm, \ouralgorithm{} (\textbf{Qu}antitative Automata \textbf{In}ference from \textbf{T}heory and \textbf{I}ncremental \textbf{C}onstraints), the learner applies deductive reasoning, using a theory, preference model, and recursive valuation function in order to infer the correct minimal automaton satisfying the currently known constraints. Search completeness of an infinite hypothesis space and minimalism of the result are guaranteed by iterative deepening, with search and inference guided by a MaxSMT objective. As such, \ouralgorithm{} is guaranteed to find the minimal automaton and terminate. In our empirical experiments, we study how \ouralgorithm{} scales, and consider ablations to the algorithm, since search completeness is expensive to maintain.

In summary, we contribute: (a) \ouralgorithm{}, an algorithm for learning deterministic quantitative automata from actively obtained comparison information, (b) theoretical guarantees of minimalism and search completeness, and (c) numerous empirical results illustrating the ability of \ouralgorithm{} to learn a variety of deterministic quantitative automata under various valuation functions.

%%%%%%%%%%%%%%%%%%%%%%%%%%%%%%%%%%%%%%%%%%%%%%%%%%%%%%%%%%%%%%%%%%%%%%%%

\section{Related Work}
Before presenting \ouralgorithm{}, we first review existing work on actively learning automata.

\textbf{Active Automata Learning}. The majority of active automata learning algorithms are adaptations of the \lstar{} algorithm~\cite{Angluin87}, applied to learning different types of automata, including deterministic finite automata~\cite{Angluin87} which model regular languages, weighted automata~\cite{Balle2015LearningWA,Bergadano1994LearningBO} in which sequences are valued according to path and transition weights including learning deterministic weighted automata~\cite{Weiss2019} for representing distributions over sequences, symbolic automata~\cite{symbolicAutomata,Argyros2018TheLO} in which predicates summarize state transitions, lattice automata~\cite{fismanlattice} which represent partial orderings over a set of sequences, and reward machines~\cite{tappler2019based,GaonB20_nonmarkovian,xu_lstar,dohmen-2022-icaps} which model certain types of non-Markovian reward functions. All of the aforementioned works utilize actively obtained concrete input-output example data.

\textbf{Automata Learning from Preferences}. In contrast, few works consider actively learning automata from \emph{non-concrete data} such as preference or constraint information. Deterministic finite automata have been learned from a combination of preference information and input-output examples~\cite{shah2023learning}, albeit using a passive automaton learning method based on state-merging, and Moore machines have been learned from preference information and feedback using \textsc{Remap}~\cite{remap}, an \lstar{}-based algorithm, featuring a symbolic observation table rather than a concrete one. Both aforementioned algorithms assume that the entire space of input sequences can be partitioned into distinct sets, with each set corresponding to a class, and where an ordering exists over the set of classes. Both the binary class case~\cite{shah2023learning} and the multiclass case~\cite{remap} are handled, and both utilize the simplistic preference model of ordering sequences based on which class the sequence belongs to.

In contrast, \ouralgorithm{} learns certain types of deterministic quantitative automata, in part by utilizing preference models which depend on the valuation function used, such as summation, discounted summation, and product, and it also subsumes the classification preference model. Additionally, \ouralgorithm{} addresses the problem of handling ambiguous variable equalities, a situation which never occurs under the classification preference models of the other two algorithms~\cite{shah2023learning,remap}, since variable equivalence or inequivalence can always be determined from a classification preference model. Algorithmically, \ouralgorithm{} is similar to \textsc{Remap} as an \lstar{}-based algorithm, since both rely on determining variable equivalences from constraints obtained from preference and equivalence queries in order to hypothesize a deterministic automata. However, due to the possibility of variable equivalence ambiguity, \ouralgorithm{} cannot apply the greedy unification of \textsc{Remap}, so instead, \ouralgorithm{} incorporates backtracking and search in the algorithm. Prior to delving into \ouralgorithm{}, we first review some notation and concepts used in the algorithm.

%%%%%%%%%%%%%%%%%%%%%%%%%%%%%%%%%%%%%%%%%%%%%%%%%%%%%%%%%%%%%%%%%%%%%%%%

\section{Background}
We first provide some background on notation and sequences, then describe the definition of deterministic quantitative automata used in this paper. We then cover core concepts which are critical to understanding \ouralgorithm{}.

\textbf{Indicator Function}: We denote the indicator function by $\mathbf{1}_{[\phi]}$, which resolves to $1$ if the statement $\phi$ resolves to true, and $0$ if $\phi$ resolves to false. For example, $\mathbf{1}_{[x\in\mathbb{R}\land x^2 < 1]}$ resolves to $1$ if $-1 < x < 1$ and $x$ is a real number, but in all other cases resolves to $0$.

\textbf{Equivalence, Equivalence Classes, and Representatives}: We use $\equiv$ to denote equivalence between a pair of objects. We write $a\equiv b$ to indicate that objects $a$ and $b$ are equivalent. An \emph{equivalence class} is set where all elements are equivalent to each other. A single element called a \emph{representative} can be used to represent the equivalence class. We denote $\repr{x}$ to be the representative of element $x$. For example, let $\{x_0,x_1,x_2,x_3,x_4\}$ be an equivalence class of variables, and let $x_0$ be the representative. Then $\repr{x_k}=x_0$ for $k=0,1,2,3,4$. Additionally, we state that two row vectors $\mathbf{r_1}=\sum_{i=1}^{d} a_i\hat{\mathbf{e}}_i$ and $\mathbf{r_2}=\sum_{i=1}^{d} b_i\hat{\mathbf{e}}_i$, with basis vectors $\hat{\mathbf{e}}_1,\dots,\hat{\mathbf{e}}_d$ in $\mathbb{R}^d$, are equivalent if the representatives of each of their components is the same. That is, $\mathbf{r_1}\equiv \mathbf{r_2}$ if $\repr{a_i}=\repr{b_i}$ for $i=1,\dots,d$.

\textbf{Alphabets and Sequences}: An alphabet is a set of elements from which sequences can be constructed. If $\Sigma$ is an alphabet, then we let $(\Sigma)^k$ represent the set of all sequences of length $k$ constructed using elements of $\Sigma$. We represent sequences constructed from $\Sigma$ of any length at least $0$ by $\Sigma^*=\bigcup_{i=0}^\infty(\Sigma)^i$. A sequence of length $0$ is denoted by $\varepsilon$. The length of a sequence $s$ is given by $|s|$. Concatenating a sequence $b$ to a sequence $a$ is denoted by $a\cdot b$, and $|a\cdot b| = |a|+|b|$. We often refer to two alphabets in this paper: an input alphabet, denoted by $\Sigma^I$, and an output alphabet, denoted by $\Sigma^O$.

\textbf{Deterministic Quantitative Automaton}: Given a set of states $Q$, initial state $q_0$, finite input and output alphabets $\Sigma^I$ and $\Sigma^O$, transition function $\delta: Q\times \Sigma^I\rightarrow Q$, labeling function $L:Q\rightarrow \Sigma^O$, and value function $\mathbf{Val}:\mathbb{R}^*\rightarrow\mathbb{R}$, a deterministic quantitative automaton is the tuple $\mathcal{A}=\langle Q,q_0,\Sigma^I,\Sigma^O,\delta,L,\mathbf{Val}\rangle$. Given an input sequence $s\in (\Sigma^I)^k$ of length $k\geq 1$, an output sequence $t\in(\Sigma^O)^{k+1}$ is generated according to the sequence of labeled states that $s$ transitions through. That is, $t=[L(q_0),L(\delta(q_0,s_{:1})),\cdots,L(\delta(q_0,s_{:k}))]$, where $s_{:j}$ represents the $j$-length prefix of $s$, and $\delta: Q\times(\Sigma^I)^*\rightarrow Q$ is the extended transition function. The valuation function $\mathbf{Val}$ then maps $t$ to a number. For notational convenience, we overload to $\mathbf{Val}:(\Sigma^I)^*\rightarrow\mathbb{R}$. Therefore, a deterministic quantitative automata $\mathcal{A}$ maps input sequences to real numbers, and we write $\mathcal{A}(s)=\mathbf{Val}(s)=v$ to indicate that $v\in\mathbb{R}$ is the value of input sequence $s$ when processed by $\mathcal{A}$, or likewise by the valuation function.

\textbf{Symbolic Observation Table}. A \emph{symbolic observation table} $\mathcal{O}$ is a 2-dimensional array of variables, denoted by $\obstable$, with rows indexed by elements of the prefix set $S$ and columns indexed by elements of the suffix set $E$, and with entries given by $T$, the empirical observation function $T: (S\cup (S\cdot\Sigma^I))\cdot E\rightarrow \Gamma$ which maps a concatenation of prefix and suffix to a variable in the context $\Gamma$. The context $\Gamma$ is a mapping from sequences to variables. If $\Gamma[t]$ is the variable associated with sequence $t$, then for all prefix-suffix $s\cdot e$ combinations in $\mathcal{O}$ for which $t = s\cdot e$, the relation $T(s\cdot e)=\Gamma[t]$ holds. The set $\mathcal{C}$ represents a set of collected constraints expressed over the variables in $\Gamma$. Figure \ref{fig:lstar-remap} depicts an example symbolic observation table, along with table operations discussed next. A symbolic observation table exhibiting certain properties can be used to construct the states $Q$, initial state $q_0$, transition function $\delta: Q\times\Sigma^I\rightarrow Q$, and labeling function $L: Q\rightarrow \Sigma^O$ of a deterministic quantitative automaton. The required properties are the following:
\begin{definition}[Unified] A symbolic observation table $\mathcal{O}=\obstable$ is \textbf{unified} if for all $s\cdot e\in (S\cdot(S\cdot\Sigma^I))\cdot E$ it holds that $T(s\cdot e)=\repr{\Gamma[s\cdot e]}$. That is, the table is populated using only the representatives of variables.
\end{definition}
\begin{definition}[Closed] Let $\mathcal{O}=\obstable$ be a symbolic observation table. Let $\orow{s}$ be the row in $\mathcal{O}$ indexed by the prefix $s$. Define the set of rows $\orows{S}=\{\orow{s}|s\in S\}$ to be the set of all the rows indexed by prefixes in $S$. The symbolic observation table $\mathcal{O}$ is \textbf{closed} if $\orows{S\cdot\Sigma^I}\subseteq\orow{S}$.
\end{definition}
\begin{definition}[Consistent] A symbolic observation table $\mathcal{O}=\obstable$ is \textbf{consistent} if for all prefix pairs $(s_1, s_2)$ in $S\times S$ where $\orow{s_1}\equiv\orow{s_2}$, it holds that all their transitions are equivalent to each other: for all $\sigma\in\Sigma^I$, $\orow{s_1\cdot \sigma}\equiv \orow{s_2\cdot\sigma}$.
\end{definition}
If a table is unified, closed, and consistent, then the transition function $\delta: Q\times\Sigma^I\rightarrow Q$ can be constructed, since consistency corresponds with a deterministic transition function, and closedness means that $Q$ is a closed set of states. The construction is the following:
\begin{align}\label{eq:states}
    Q &=\{\orow{s} | \text{ for all } s\in S\} \text{ (set of states)}\\
    q_0&=\orow{\varepsilon} \text{ (initial state)}\\
    \delta(\orow{s},\sigma)&=\orow{s\cdot \sigma} \text{ for all } s\in S \text{ and } \sigma\in\Sigma^I \text{ (transition function)}\\
    \tilde{L}(\orow{s})&=T(s\cdot\varepsilon) \text{ (symbolic output function)}\\
    \Lambda &: \{\Gamma[s\cdot e]\mapsto r_{s\cdot e} | r_{s\cdot e}\in \Sigma^O\} \text{ (solution or model satisfying } \mathcal{C}\text{)}\\
    L(\orow{s})&=\Lambda[T(s\cdot \varepsilon)] \text{ (concrete output function)}
\end{align} Note that $\Lambda$ is an assignment of values of $\Sigma^O$ to the variables in the context $\Gamma$ such that the assignment satisfies the constraints in $\mathcal{C}$.
\begin{figure}[t]
    \centering
    \begin{tabular}{c|c}
     \includegraphics[width=0.49\textwidth]{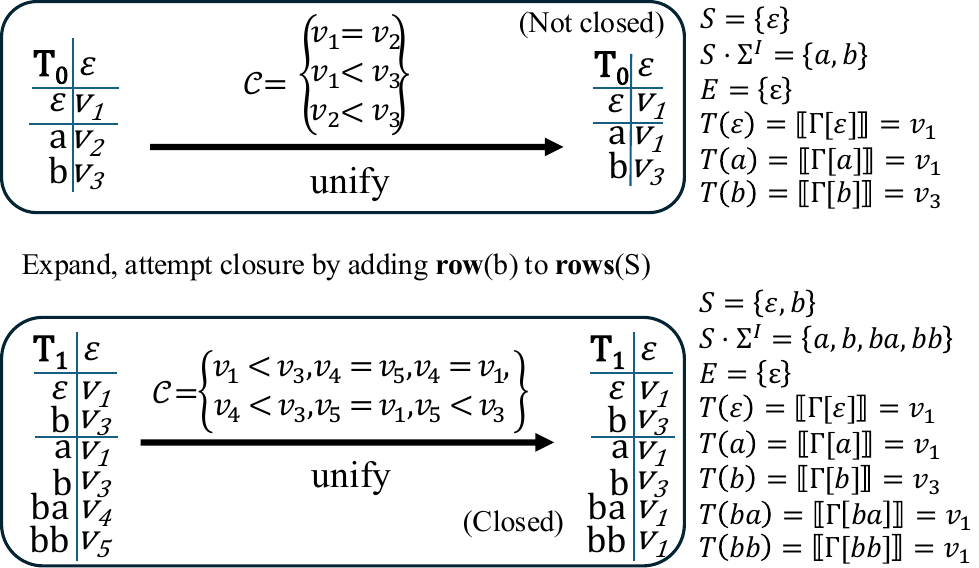} & \includegraphics[width=0.49\textwidth]{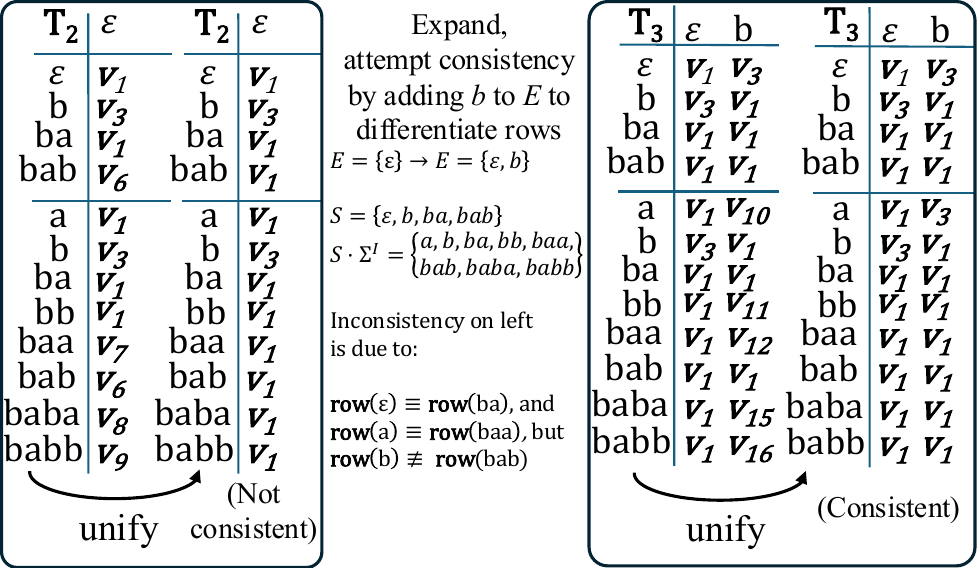}\\\includegraphics[width=0.49\textwidth]{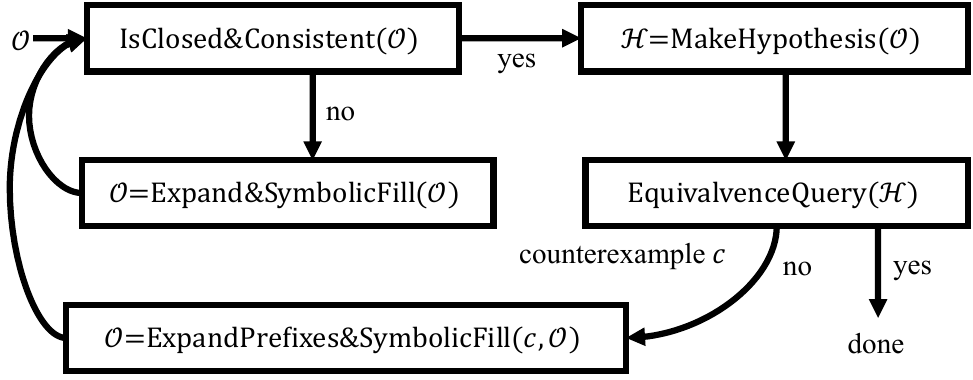} & \includegraphics[width=0.49\textwidth]{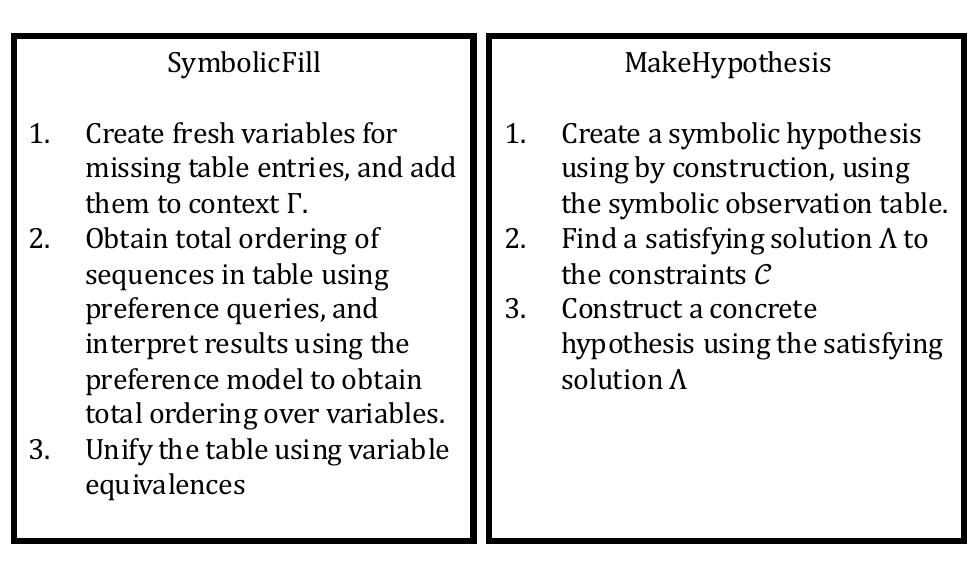}\\
    \end{tabular}
    \caption{Upper left: closure operation on symbolic observation table. Upper right: consistency operation on symbolic observation table. Bottom Left: high-level \textsc{Remap} algorithm, which is an algorithm based on \lstar{}. Bottom Right: Details for \textsc{Remap} \textsc{SymbolicFill} and \textsc{MakeHypothesis} procedures.}
    \label{fig:lstar-remap}
\end{figure}

\textbf{Formal Theory}. A \emph{formal theory} is the tuple $\langle \Xi, \alpha, \mathcal{F}\rangle$, where \emph{signature} $\Xi$ is a set of constant, function, and predicate symbols, \emph{axioms} $\alpha$ is a finite set of closed formulae utilizing only symbols from $\Xi$, and $\mathcal{F}$ is the set of all statements utilizing symbols from $\Xi$, as well as variables, logical connectives, and quantifiers, such that every model $\Lambda$ satisfying a formula $F\in \mathcal{F}$ also satisfies all axioms in $\alpha$. In other words, a formal theory encompasses a set of axioms, as well as all statements which can be logically derived from the axioms.

\textbf{The \lstar{} Algorithm}. \lstar{} \cite{Angluin87} is an active learning algorithm for learning automata, wherein the learner $\mathcal{L}$ must output an automaton $\mathcal{H}$ equivalent to the target automaton $\mathcal{A}=\langle Q,q_0,\Sigma^I,\Sigma^O,\delta,L\rangle$ the teacher $\mathcal{T}$ has in mind. Specifically, the learner must hypothesize an automaton $\mathcal{H}$ such that $\mathcal{H}(s)=\mathcal{A}(s)$ for all $s\in(\Sigma^I)^*$, where $\mathcal{A}(s)=L(\delta(q_0,s))$. While $\mathcal{L}$ and $\mathcal{T}$ share knowledge of alphabets $\Sigma^I$ and $\Sigma^O$ with $|\Sigma^O|=2$, the learner can only obtain information from $\mathcal{T}$ by asking two types of queries: \emph{membership queries} $\mathbf{memQ}(s)=\mathcal{A}(s)$ provide an observation of the concrete value used to label the state $\delta(q_0,s)$, and \emph{equivalence queries} $\mathbf{equivQ}(\mathcal{H})=\langle \forall s: \mathcal{H}(s)=\mathcal{A}(s), c: \mathcal{H}(c)\neq\mathcal{A}(c) \rangle$, where $\mathcal{T}$ checks if there exists a counterexample $c$ for which $\mathcal{H}(c)\neq\mathcal{A}(c)$ and returns $\langle \textbf{False}, c\rangle$ if $c$ exists; otherwise if $c$ does not exist, $\mathcal{T}$ returns $\langle \mathbf{True}, \cdot\rangle$. The algorithm utilizes a concrete observation table and proceeds with $\mathcal{L}$ asking membership queries and expanding the table until it becomes closed and consistent, after which it constructs a hypothesis $\mathcal{H}$ to ask in an equivalence query. If a counterexample $c$ is received, the table is expanded with $c$ and all its prefixes. This entire process repeats until no counterexample is returned from the equivalence query.

\textbf{The \textsc{Remap} Algorithm}. \textsc{Remap}~\cite{remap} is an \lstar{}-based algorithm, where preference queries $\mathbf{prefQ}(s_1,s_2)$ replace membership queries, and $\mathcal{T}$ returns one of $\mathcal{A}(s_1)>\mathcal{A}(s_2)$, $\mathcal{A}(s_1) <\mathcal{A}(s_2)$, or $\mathcal{A}(s_1)=\mathcal{A}(s_2)$, depending on its preference of $s_1\succ s_2$, $s_1\prec s_2$, or $s_1\sim s_2$~respectively, which the learner encodes as $\Gamma[s_1]>\Gamma[s_2],\Gamma[s_1]<\Gamma[s_2]$, or $\Gamma[s_1]=\Gamma[s_2]$ into $\mathcal{C}$. The equivalence query is modified to a 3-tuple which includes feedback on the value of $c$ via $\mathbf{equivQ}(\mathcal{H})=\langle \forall s: \mathcal{H}(s)=\mathcal{A}(s), c: \mathcal{H}(c)\neq\mathcal{A}(c), \mathcal{A}(c) \rangle$. Figure \ref{fig:lstar-remap} depicts the \textsc{Remap} algorithm at a high-level. It is assumed that $\Sigma^O$ has a total ordering, and $(\Sigma^I)^*$ has a total ordering, and that the preferences of $\mathcal{T}$ are consistent with both. The algorithm proceeds similarly to \lstar{}, performing preference queries and table expansions in order to obtain a unified, closed, and consistent observation table so that it can ask an equivalence query. Counterexamples are processed identically to \lstar{}. However, \textsc{Remap} operates symbolically, and relies on creating equivalence classes of variables from known variable equalities in $\mathcal{C}$. Conjecturing satisfying variable equalities from constraints turns out to be a pivotal step in \ouralgorithm{}. We now detail the problem formulation \ouralgorithm{} solves.

%%%%%%%%%%%%%%%%%%%%%%%%%%%%%%%%%%%%%%%%%%%%%%%%%%%%%%%%%%%%%%%%%%%%%%%%
\begin{figure}[b!]
    \centering
    \begin{tabular}{r|l}
        \includegraphics[width=0.49\textwidth]{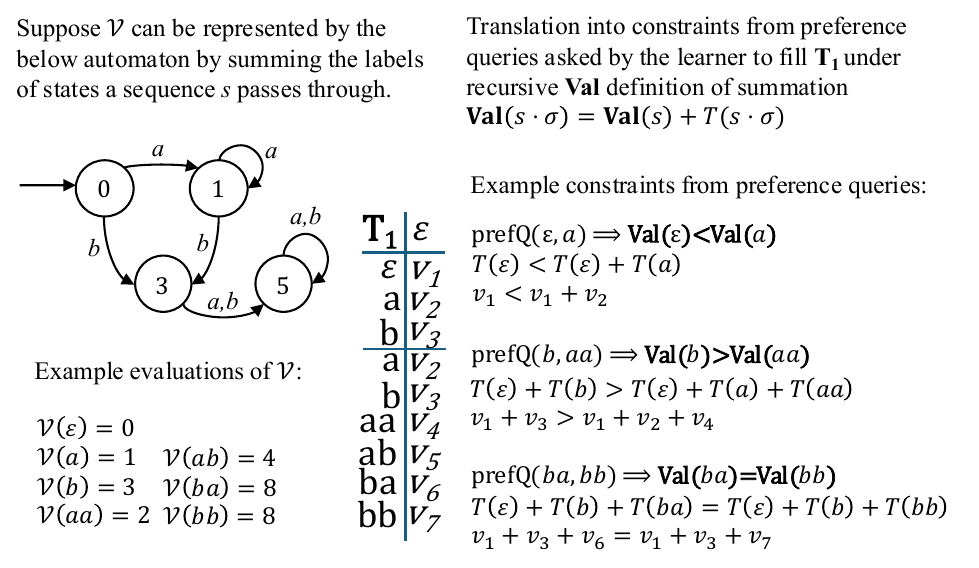} & \includegraphics[width=0.49\textwidth]{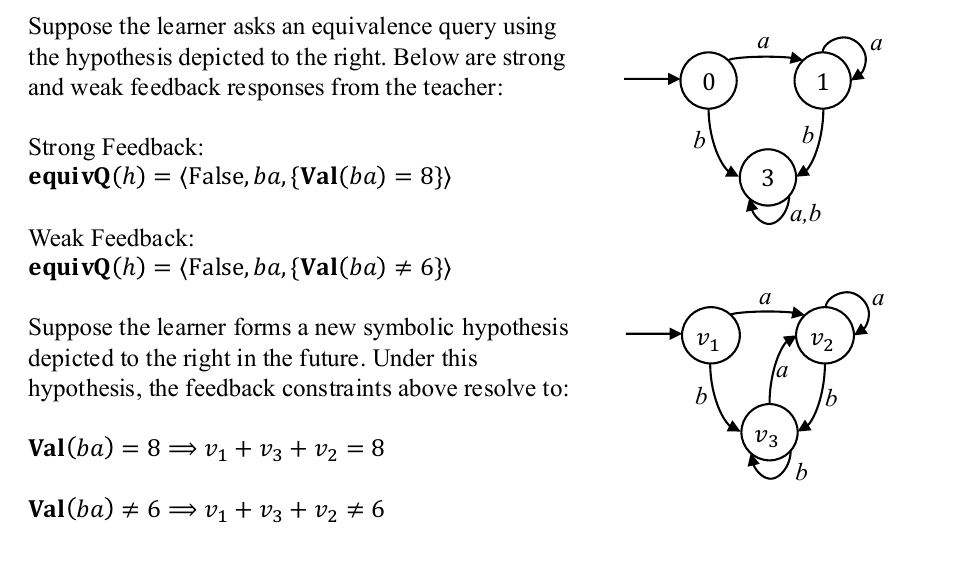}
    \end{tabular}
    \caption{Example constraints obtained under \ouralgorithm{} for preference queries (left) and equivalence queries (right) under non-discounted summation $\mathbf{Val}$ and $\mathcal{V}$ functions.}
    \label{fig:prefq-equivq-example}
\end{figure}
\section{Problem Formulation}

Consider an active learning setup where the learner $\mathcal{L}$ must learn a function $\mathcal{V}: (\Sigma^I)^*\rightarrow \mathbb{Q}$ exactly as a deterministic quantitative automaton $\mathcal{A}=\langle Q,q_0,\Sigma^I,\Sigma^O,\delta,L,\mathbf{Val}\rangle$. The learner can only obtain information about $\mathcal{V}$ by asking the teacher $\mathcal{T}$ preference and equivalence queries. The learner is aware of the algebraic form of $\mathcal{V}$, and sets $\mathbf{Val}$ to have equivalent form. The teacher is assumed to have preferences consistent with $\mathcal{V}$, and responds to the preference query as follows:
\begin{align}
    \textbf{prefQ}(s_1,s_2) = \mathcal{P}(s_1,s_2|\mathcal{V})=\left\{\begin{array}{rl}
    1 & \text{ if } \mathcal{V}(s_1) > \mathcal{V}(s_2)\\
    0 & \text{ if } \mathcal{V}(s_1) = \mathcal{V}(s_2)\\
    -1& \text{ if } \mathcal{V}(s_1) < \mathcal{V}(s_2)
\end{array}\right.
\end{align} For equivalence queries, the teacher responds with a 3-tuple as follows:
\begin{align}
    \mathbf{equivQ}(\mathcal{A})=\langle \forall s: \mathcal{A}(s)=\mathcal{V}(s), c: \mathcal{A}(c)\neq\mathcal{V}(c), f(c,\mathcal{A},\mathcal{V})\rangle
\end{align} where feedback is represented by $f(c,\mathcal{A},\mathcal{V})$. If the provided feedback is \emph{strong}, then the learner acquires the constraint $\mathbf{Val}(c)=\mathcal{V}(c)$ from $f(c,\mathcal{A},\mathcal{V})$, indicating that the concrete value of $c$ is $\mathcal{V}(c)$. If the feedback is \emph{weak}, then the learner obtains $\mathbf{Val}(c)\neq\mathcal{A}(c)$ from $f(c,\mathcal{A},\mathcal{V})$, indicating that the concrete value of $c$ is not $\mathcal{A}(c)$ for the current hypothesis $\mathcal{A}$. Figure \ref{fig:prefq-equivq-example} depicts examples constraints obtained from preference queries and equivalence queries under non-discounted summation $\mathbf{Val}$ and $\mathcal{V}$ functions.

Thus, given a teacher $\mathcal{T}$ capable of answering \textbf{prefQ} and \textbf{equivQ} asked by a learner $\mathcal{L}$, the learner must infer the minimal deterministic quantitative automaton $\mathcal{A}$ equivalent to $\mathcal{V}$. We now describe how \ouralgorithm{} solves this problem.

\begin{algorithm}[t!]
\caption{\textbf{Qu}antitative Automata \textbf{In}ference from \textbf{T}heory and \textbf{I}ncremental \textbf{C}onstraints}
\label{alg:quintic-pseudo}
\textbf{Input}: Alphabets $\Sigma^I$ (input) and $\Sigma^O$ (output), teacher $\mathcal{T}$ utilizing valuation function $\mathcal{V}$\\
\textbf{Output}: Deterministic Quantitative Automaton $\mathcal{A} = \langle \hat{Q}, \Sigma^I, \Sigma^O, \hat{q}_0, \hat{\delta}, \hat{L}, \mathbf{Val}\rangle$\\
\textbf{Assumptions}: A formal theory $\mathbf{T}$ for which there exists an $\mathcal{R}\in\{=,\neq\}$ such that $\mathcal{V}(s)\mathbf{R}\mathbf{Val}(s)$ is valid in $\mathbf{T}$. A recursive definition of $\mathbf{Val}$ exists. $\mathcal{L}$ is a stack, where each element is a pair of observation table and equivalence class set. $\mathcal{S}$ is an SMT solver.\\
\textbf{Remarks:} Due to the recursive definition of $\mathbf{Val}$, the set $(S\cup(S\cdot\Sigma^I))\cdot E$ must be prefix-closed. This means $E$ must now also be prefix-closed: during a consistency fix operation, if $e$ is added to $E$, then all prefixes of $e$ are also added to $E$.\\
\begin{algorithmic}[1]
\STATE Initialize $\mathcal{O} = \obstable$ with $S=\{\varepsilon\},E=\{\varepsilon\}$, $\mathcal{C}=\{\}$, $\Gamma=\emptyset$\\
\STATE \textit{sat?}$=$\textit{UNSAT}, \textit{iscorrect} $=$ False, $\mathcal{L}=$ \textsc{Stack}(), $\mathcal{S}=$ \textsc{Solver}()
\STATE $\mathcal{O,L,S}\longleftarrow$ \textsc{SymbolicFill}$\left(\mathcal{O,L,S,T},\Sigma^I, \Sigma^O,\text{backtrack=False}\right)$
\REPEAT
\WHILE{\textit{sat?} is \textit{UNSAT}}
\STATE $\mathcal{O,L,S} \longleftarrow$ \textsc{MakeClosedAndConsistent}$\left(\mathcal{O,L,S,T},\Sigma^I, \Sigma^O\right)$
\STATE \textit{sat?}, $\mathcal{H}\longleftarrow$\textsc{MakeSymbolicHypothesis}$\left(\mathcal{O,S}, \Sigma^I, \Sigma^O\right)$
\IF{\textit{sat?} is \textit{UNSAT}}
    \STATE $\mathcal{O,L,S}\longleftarrow$ \textsc{SymbolicFill}$\left(\mathcal{O,L,S,T},\Sigma^I, \Sigma^O,\text{backtrack=True}\right)$
\ENDIF
\ENDWHILE
\STATE $\mathcal{S}$.\textsc{Push}()
\STATE Exhaustively test all concrete hypotheses $h$ under symbolic hypothesis $\mathcal{H}$ and constraints $\mathcal{C}$ using \textsc{EquivQ}$(h)$ and record counterexamples and symbolic constraints into $\mathcal{C}$ until an $h$ is found to be correct or there exist no more valid $h$ under $\mathcal{H}$ and the updated constraints $\mathcal{C}$.
\STATE $\mathcal{S}$.\textsc{Pop}()
\IF{$h$ is correct}
\RETURN $h$
\ELSE
\STATE $\mathcal{O,L,S}\longleftarrow$ \textsc{SymbolicFill}$\left(\mathcal{O,L,S,T},\Sigma^I, \Sigma^O,\text{backtrack=True}\right)$
\STATE Expand $\mathcal{O}$ with prefixes of all counterexamples obtained from Line 13
\STATE $\mathcal{O,L,S}\longleftarrow$ \textsc{SymbolicFill}$\left(\mathcal{O,L,S,T},\Sigma^I, \Sigma^O,\text{backtrack=False}\right)$
\ENDIF
\UNTIL{\textit{iscorrect} $=$ True}
\end{algorithmic}
\end{algorithm}

\section{\ouralgorithm{}}

\ouralgorithm{}, shown in Algorithm \ref{alg:quintic-pseudo}, is an \lstar{}-based algorithm \cite{Angluin87} for actively learning quantitative automata from constraints collected through preference and equivalence queries. Given a pair $\mathbf{Val}$ and $\mathcal{V}$ such that $\mathcal{V}(s)\mathbf{R}\mathbf{Val}(s)$ is always a statement contained in a formal theory $\mathbf{T}$ for a relation $\mathbf{R}\in\{=,\neq\}$, and a set of collected constraints $\mathcal{C}$ contained in $\mathbf{T}$, \ouralgorithm{} applies deductive reasoning using $\mathbf{T}$, as well as an informed search encoded as a maximum satisfiability (MaxSMT) problem to arrive at the correct deterministic quantitative automaton. Iterative deepening depth first search (IDS) guarantees completeness and minimalism, since the learner is unaware of the minimum number of states required to represent $\mathcal{V}$.

\ouralgorithm{}, similar to \textsc{Remap} \cite{remap}, employs a symbolic observation table $\mathcal{O} =\obstable$, and executes requisite preference queries and table expansions in order to obtain a unified, closed, and consistent table, from which an automaton hypothesis can be constructed.

The constraints obtained from preference queries are determined by $\mathcal{V}$ and $\mathbf{Val}$ via $\mathcal{V}(s_1)\,\mathbf{R}\,\mathcal{V}(s_2)\implies\mathbf{Val}(s_1)\,\mathbf{R}\,\mathbf{Val}(s_2)$ for the relation $\mathbf{R}\in\{>,<,=\}$. In order to apply deductive reasoning to arrive at a valid quantitative automaton, a suitable theory must also be supplied to the learner. We consider four forms of the valuation function, based on the automaton type. Below, we illustrate the corresponding forms for $\mathcal{A}(s)$ and the recursive form of $\mathbf{Val}$, where $T$ is the empirical observation function from the observation table:
\begin{itemize}
    \item Non-Discounted Summation $\mathcal{A}(s)=\sum_{i=0}^{|s|} L(\delta(q_0,s_{:i}))$ and $\mathbf{Val}(s\cdot\sigma) = T(s\cdot\sigma) + \mathbf{Val}(s)$ with $\mathbf{Val}(\varepsilon)=T(\varepsilon)=0$ 
    \item Discounted Summation $\mathcal{A}(s)=\sum_{i=0}^{|s|} \gamma^iL(\delta(q_0,s_{:i}))$ and\\$\mathbf{Val}(s\cdot\sigma)=\gamma^{|s\cdot\sigma|}T(s\cdot\sigma)+\mathbf{Val}(s)$ with $\mathbf{Val}(\varepsilon)=T(\varepsilon)=0$ 
    \item Product $\mathcal{A}(s)=\prod_{i=0}^{|s|} L(\delta(q_0,s_{:i}))$ and\\$\mathbf{Val}(s\cdot\sigma)=T(s\cdot\sigma)\mathbf{Val}(s)$ with $\mathbf{Val}(\varepsilon)=T(\varepsilon)=1$
    \item Classification: $\mathcal{A}(s)=L(\delta(q_0,s))$ and $\mathbf{Val}(s)=T(s)$ 
\end{itemize}
Determining whether a pair of variables in the table are equivalent is a critical operation in \ouralgorithm{}, since equivalence underscores closed and consistency checks, and this must be accomplished for each \textbf{Val} supported by \ouralgorithm{}. Conjecturing variable equivalences occurs within the \textsc{SymbolicFill} call, so we organize our discussion of \ouralgorithm{} on how variable equivalences are conjectured, and how a complete search over chains of alternating \textsc{SymbolicFill} calls and table expansions occurs in order to arrive at the correct hypothesis.

\subsection{Inferring Variable Equivalences}
Given a recursive form of $\mathbf{Val}$, inference rules can be derived to determine certain variable equivalences. Importantly, variable equivalence is not decidable for all forms of $\mathbf{Val}$. By this, we mean that the learner simply making a requisite number of preference queries is not guaranteed to determine variable equivalence from preference constraints alone. Below, we present the inference rules required each \textbf{Val} definition. In particular, we require inference rules which map from a set preference queries to a conclusion \begin{equation}\label{eq:2}T(s\cdot\sigma)~\mathbf{R}~T(s'\cdot\sigma')\text{ where }\mathbf{R}\in\{<,>,=\}.\end{equation} The inference rules can be obtained by writing Equation \ref{eq:2} in terms of $\mathbf{Val}$, and considering what combinations of preference queries would satisfy Equation \ref{eq:2}. As a result, the inference rules can be concisely encoded into a decision tree with inputs as specific preference queries. For succinctness, we denote the following when encoding the decision tree:
\begin{align*}
(X,Y)&=(\mathbf{prefQ}(s\cdot\sigma,s'\cdot\sigma'), \mathbf{prefQ}(s,s'))\\
(X',Y')&=(\mathbf{prefQ}(s\cdot\sigma, s), \mathbf{prefQ}(s'\cdot\sigma', s'))\\
z(x,y)&=\mathbf{1}_{[x>y]}-\mathbf{1}_{[x<y]} \text{ for } x,y\in\{1,0,-1\}\\
Z(x,y)&=\left\{\begin{array}{rl}
    > & \text{ if } z(x,y)=1\\
    = & \text{ if } z(x,y)=0\\
    < & \text{ if } z(x,y)=-1
\end{array}\right.
\end{align*}
\begin{itemize}
    \item Non-Discounted Summation: If $|X+Y|=2$ and $|X'+Y'|=2$, then the relation $\mathbf{R}$ in $T(s\cdot\sigma)\mathbf{R}T(s'\cdot\sigma')$ cannot be determined. Otherwise, if $|X+Y|=2$ and $|X'+Y'|\neq2$, then $\mathbf{R}=Z(X',Y')$. Otherwise, if $|X+Y|\neq2$, then $\mathbf{R}=Z(X,Y)$.
    \item Product: Identical to non-discounted summation.
    \item Discounted Summation: If $|s|=|s'|$, then use non-discounted summation. Otherwise, if $|s|\neq|s'|$ and $|X'+Y'|=2$, then the relation $\mathbf{R}$ in $T(s\cdot\sigma)\mathbf{R}T(s'\cdot\sigma')$ cannot be determined. Otherwise, if $|s|\neq|s'|$ and $|X'+Y'|\neq2$, then $\mathbf{R}=Z(X',Y')$.
    \item Classification: $\mathbf{R}$ is determined directly from $\mathbf{prefQ}(s\cdot\sigma,s'\cdot\sigma')$, so no inference is needed.
    
\end{itemize}

\begin{figure}[b]
    \centering
    \includegraphics[height=0.25\textwidth]{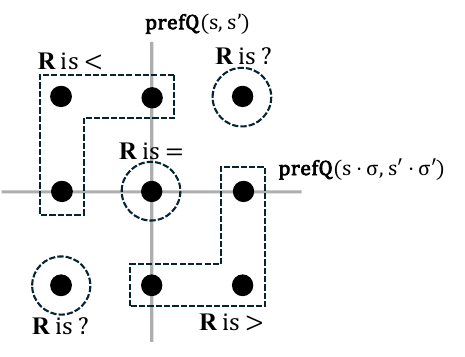}
    \includegraphics[height=0.25\textwidth]{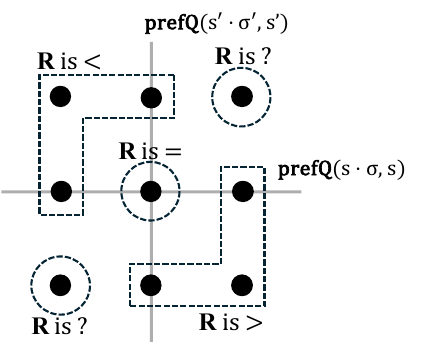}
    \caption{Variable equivalence inference rules encoded as a decision tree. Preference query pairs indicate the relation $\mathbf{R}$ in $T(s\cdot\sigma) \mathbf{R} T(s'\cdot\sigma') \iff \mathbf{Val}(s\cdot \sigma)+\mathbf{Val}(s') \mathbf{R} \mathbf{Val}(s'\cdot \sigma') + \mathbf{Val}(s)$ for non-discounted summation valuation. For $\mathbf{R}$ to remain unknown, both preference query pairs must reflect that $\mathbf{R}$ is unknown. Otherwise, the relation $\mathbf{R}$ can be determined.}
    \label{fig:decision-tree-inference-rules}
\end{figure}

Figure~\ref{fig:decision-tree-inference-rules} depicts the decision tree regions, which are points in $\{-1,0,1\}\times\{-1,0,1\}$ space. Note that the relation $\mathbf{R}$ in $T(s\cdot\sigma)\mathbf{R}T(s'\cdot\sigma')$ cannot always be determined, and therefore, unification cannot be greedily applied to replace all members of each equivalence class with their representatives. However, in order to make a hypothesis, the learner must ensure closedness and consistency, which entails determining whether pairs of variables in the table are equivalent. Since there may be pairs of variables for which equivalence is unknown from preference constraints alone, the learner must \emph{conjecture} variable equivalences and test whether each sequence of conjectures it makes leads to a correct hypothesis or not. Given a set of constraints $\mathcal{C}$, let $\mathcal{E}$ represent a set of variable equivalences, inclusive of known equivalences derivable from inference rules and conjectured equivalences. We call $\mathcal{E}$ a \emph{set of equivalence classes}. Since \ouralgorithm{} may need to repeatedly make conjectures, it is helpful to represent Algorithm~\ref{alg:quintic-pseudo} as repeatedly attempting to satisfy an SMT problem.

Thus, the problem of finding a satisfying automaton $\mathcal{H}$ can be expressed as repeatedly solving a satisfaction problem for a given symbolic observation table $\mathcal{O}$, set of preference constraints $\mathcal{C}$, and set of feedback constraints $\mathcal{B}$. At each iteration, we ask, does there exist a pair $(\mathcal{E,H})$ satisfying
\begin{equation}\label{eq:1}
    \mathcal{E} = f_1(\mathcal{O,C})\land \mathcal{H}=f_2(\mathcal{E,O})\land f_3(\mathcal{H,B})
\end{equation} 

\ouralgorithm{} operates as follows: starting with an initial $\mathcal{O,C,B}$, determine if Equation \ref{eq:1} is satisfiable. If not, change $\mathcal{O}$ by expanding it, and as a result, collect more constraints into $\mathcal{C}$ and $\mathcal{B}$ via preference queries and equivalence queries, and try again. Observe that this process of repeatedly attempting to satisfy Equation \ref{eq:1} corresponds to Algorithm \ref{alg:quintic-pseudo}, with $\mathcal{E}=f_1(\mathcal{O,C})$ representing variable equivalence conjectures inside \textsc{SymbolicFill}, term $\mathcal{H}=f_2(\mathcal{E,O})$ representing Lines 5--11, and term $f_3(\mathcal{H,B})$ representing Line~13. Note that the expansion and querying procedure follows closely to REMAP and \lstar{}, but is more generalized to accommodate valuations other than classification. We will now go over each of the terms in Equation~\ref{eq:1}.

\subsection{Conjecturing Variable Equivalences within \textsc{SymbolicFill}}

The clause $\mathcal{E}=f_1(\mathcal{O,C})$ corresponds to solving for a set of equivalence classes that satisfy $\mathcal{C}$, given a symbolic observation table~$\mathcal{O}$. As described previously, certain variable equivalences cannot be determined from preference query constraints in $\mathcal{C}$ alone. 

Given a set of obtained constraints $\mathcal{C}$ expressed over $m$ variables in table $\mathcal{O}$, and a theory $\mathbf{T}$ appropriate for inferring variable equalities from constraints in $\mathcal{C}$, the learner applies appropriate inference rules from $\mathbf{T}$ to infer a set of variable pairs $\mathcal{K}$ with \emph{known} equivalence or inequivalence relations, and a set of variable pairs $\mathcal{U}$ with \emph{unknown} relations. That is,
\begin{equation}\label{eq:known-unknown}(\mathcal{C},\mathbf{T})\implies (\mathcal{K,U}) \text{ where } |\mathcal{K}|+|\mathcal{U}|=m(m-1)\end{equation}
Once $\mathcal{K}$ and $\mathcal{U}$ have been obtained, the learner must hypothesize possible equivalences for pairs from $\mathcal{U}$ in order to create a possible set of equivalence classes $\mathcal{E}$ satisfying $\mathcal{C}$.

Since the aim is to obtain a \emph{minimal} quantitative automaton (one with as few states as possible), we seek to maximize the number of equivalences in $\mathcal{U}$. Therefore, we consider maximizing this objective while solving for a valid solution to the constraints $\mathcal{C}$:
\begin{equation}\label{eq:max-smt}\Lambda=\argmax_{(v_1,\dots,v_m)} \displaystyle\sum_{(v_i,v_j)\in\mathcal{U}}\mathbf{1}_{\left[v_i\mathbf{R}_{ij}v_j \land \mathbf{R}_{ij}\in\{=\}\right]} \text{ s.t. } \Lambda \text{ satisfies } \mathcal{C}\end{equation}

This objective dictates the order in which feasible solutions to $\mathcal{C}$ are chosen, greedily guiding \ouralgorithm{} towards a minimal automaton. Equation~\ref{eq:max-smt} is simple, and it \emph{roughly} corresponds with maximizing the number of equivalent rows in the table, while also assigning an equivalence relation between variables in different columns. Of course, the minimalism guarantee of \ouralgorithm{} is independent from the order in which feasible solutions are selected; iterative deepening search in Section~\ref{section:chains} is what provides the guarantee.

Nonetheless, to select an $\mathcal{E}$, we note that a given solution $\Lambda$ to constraints $\mathcal{C}$ corresponds to a particular equivalence class set $\mathcal{E}$. In particular, we observe that we are assigning $m$ variables to at most $|\Sigma^O|$ buckets; all variables located in a given bucket are part of the same equivalence class. Once an equivalence class set $\mathcal{E}$ has been chosen, the symbolic observation table $\mathcal{O}$ can be unified in preparation for being made closed and consistent.

\subsection{Constrained Hypothesis Search over Conjecture--Expansion Chains}\label{section:chains}
The clause $\mathcal{H}=f_2(\mathcal{E,O})$ summarizes the process of searching for a satisfying hypothesis $\mathcal{H}$, by searching for a sequence of $(\mathcal{E,O})$ equivalence class set and observation table pairs. First, we describe how the sequence of $(\mathcal{E,O})$ pairs is generated by an alternating process. Second, we explain how the search over sequences, and therefore hypotheses, is guaranteed to be complete using iterative deepening search, in order to avoid infinite length sequences. Finally, we discuss equivalence class consistency within a given sequence, and show how the consistency requirement leads to a necessary inference rule relating equivalence class sets to number of variables, and therefore describes when the table must be expanded, resulting in guaranteed minimalism when combined with iterative deepening.

First, note that in order to produce a hypothesis, \ouralgorithm{} requires a unified, closed, and consistent observation table $\mathcal{O}$. A valid equivalence class set $\mathcal{E}$ can be used to unify the table, and once unified, $\mathcal{O}$ must be made closed and consistent, potentially through table expansions. The effect of a table expansion is the addition of more variables and the collection of more constraints.

This necessarily creates an alternating process between selecting a set of equivalence classes and seeking closedness and consistency. The $f_2$ process generates a $(d+1)$-length chain $\langle(\mathcal{E}_{k+1},\mathcal{O}_k)\rangle_{0\leq k\leq d}$ of $(\mathcal{E,O})$ pairs resulting in a hypothesis $\mathcal{H}$. Thus, $f_2$ can be viewed as searching for a satisfying chain $\langle(\mathcal{E}_{k+1},\mathcal{O}_k)\rangle_{0\leq k\leq d}$ resulting in valid $\mathcal{H}$, and is summarized as alternating between the \textsc{SymbolicFill} and \textsc{Expansion} calls below:
\begin{align*}
\textsc{SymbolicFill}&\left\{\begin{array}{rcl}
\mathcal{C}_{k-1}&=&\textsc{ObtainConstraints}(\mathcal{O}_{k-1})\\
\mathcal{E}_k&=&f_1(\mathcal{O}_{k-1}, \mathcal{C}_{k-1})\\
\tilde{\mathcal{O}}_k&=&\textsc{Unification}(\mathcal{E}_k,\mathcal{O}_{k-1})\\
\end{array}\right.\\
\textsc{Expansion}&\left\{\mathcal{O}_k=\textsc{ExpandIfNotClosedAndConsistent}(\tilde{\mathcal{O}}_k)\right.
\end{align*}
There may be multiple choices of $\mathcal{E}$ which satisfy a particular $\mathcal{C}$, implying that there may exist multiple $\langle(\mathcal{E}_{k+1},\mathcal{O}_k)\rangle_{0\leq k\leq d}$ chains of varying lengths which result in different hypotheses. In order to guarantee search completeness, and to avoid potentially infinite sequence lengths, we limit the size of $d$ by using iterative deepening depth first search (IDS). This is implemented through budgeting the number of \textsc{SymbolicFill} calls allowed. Iterative deepening implies a backtracking operation must be supported to shorten the current chain by removing the latest $(\mathcal{E,O})$ pair. Backtracking occurs whenever \textsc{SymbolicFill} is called with \texttt{backtrack=True}, and backtracking also occurs whenever the conjunction of $\mathcal{E}$ in the chain can no longer satisfy $\mathcal{C}$ in the presence of counterexample feedback.
\begin{theorem}[Completeness]
    \ouralgorithm{} performs a complete search of automata space.
\end{theorem}
\begin{proof} The result follows from iterative deepening depth-first search being complete and optimal. Specifically, \ouralgorithm{} simulates iterative deepening by budgeting \textsc{SymbolicFill} calls. Let $M$ be the maximum budget, and $b$ represent the remaining budget, initialized as $M$. Each time an equivalence class set is successfully solved for, $b$ is reduced by $1$. Each time backtracking is called, $b$ is increased by $1$. Therefore, a consistent chain of equivalence class sets will be at most $M$ units long, given a starting observation table. If none of the chains results in the correct hypothesis, then all chains are removed, and $M$ is increased by $1$, and the process repeats until the correct hypothesis is found, which will be via the shortest possible chain leading to it.\qed
\end{proof}

Note that a hypothesis $\mathcal{H}$ can be made only if $\mathcal{O}_k=\tilde{\mathcal{O}}_k$, which occurs only if $\tilde{\mathcal{O}}_k$ is already closed and consistent. Once this condition is satisfied, a hypothesis can be constructed and validated using an equivalence query. Observe that if there exists $\langle(\mathcal{E}_{k+1},\mathcal{O}_k)\rangle_{0\leq k\leq d}$ of length $(d+1)$ resulting in a hypothesis, then there is a restriction on consecutively selected equivalence class sets, namely, for each $1 \leq k\leq d$, $\mathcal{E}_k$ must be consistent with $\mathcal{E}_{k-1}$. Since the table has potential to expand in size, each subsequent $\mathcal{E}$ potentially describes an increasing number of variables. We use this insight to derive a logical inference rule that enables \ouralgorithm{} to make progress towards the solution.

Observe that $\mathcal{E}_k=f_1(\mathcal{O}_{k-1},\mathcal{C}_{k-1})$ satisfies constraints $\mathcal{C}_{k-1}$, currently over $|\Gamma|=m$ variables. If $\mathcal{E}_k$ is the true solution, then it must also satisfy all future $\mathcal{C}_j$ for $j\geq k$. Note that $\mathcal{C}_j$ is expressed over more than $m$ variables. We find that if $\mathcal{E}_k$ violates a future $\mathcal{C}_j$, then it is incorrect to simply logically negate $\mathcal{E}_k$. Instead, a negation of an implication is appropriate:
\begin{equation}|\Gamma|=m\implies \neg\,\mathcal{E} \iff \mathcal{E}\implies |\Gamma|\neq m\end{equation}
This can be strengthened to
\begin{equation}\label{eqn:var-dim}|\Gamma|\leq m \implies \neg\,\mathcal{E}\iff\mathcal{E}\implies|\Gamma|>m\end{equation}
In other words, if $\mathcal{E}$ is found to be invalid involving $m$ variables under constraints in $\mathbb{Q}^m$ space, it may still be valid under $\mathbb{Q}^{m+p}$ for some $p\geq 1$, involving $m+p$ variables.

The above inference rule allows us to detect when the table must be expanded to include more variables: if all equivalence classes have been ruled out under $m$ variables, then the table must be expanded to include more than $m$ variables. This expansion allows us to guarantee minimalism when used with IDS.\begin{theorem}[Minimalism] \ouralgorithm{} results in a minimal automaton.
\end{theorem}
\begin{proof} There exists a set $\Omega$ of infinitely many tables which are isomorphic by construction to the target. \ouralgorithm{} finds the shortest chain leading to an element in $\Omega$. Since all tables in $\Omega$ are isomorphic due to the construction mechanism Equation~\ref{eq:states} (duplicate rows are collapsed into a single state), reaching any element of $\Omega$ results in minimalism. \ouralgorithm{} happens to reach the closest element to $\mathcal{O}_0$ in $\Omega$ along the shortest chain via IDS.\qed
\end{proof}

Since $\mathcal{E}$ sequences always satisfy $\mathcal{C}$ per Equation~\ref{eq:max-smt}, equivalence classes only become invalidated by counterexample feedback. Such feedback only comes into play when a hypothesis $\mathcal{H}$ induced by a set of equivalence classes $\mathcal{E}$ is found to be invalid. Hypothesis validity is determined through an equivalence query.

\subsection{Verifying Equivalence}
The last clause $f_3(\mathcal{H,B})$ represents the process of preparing an equivalence query. Recall that $\mathcal{B}$ is the set of constraints obtained from equivalence queries. The constraints in $\mathcal{B}$ are necessarily expressed in a high-level specification in the form of $\mathbf{Val}(c)\mathbf{R}\,v$ where $c$ is a counterexample sequence, $\mathbf{R}\in\{=,\neq\}$, and $v$ is a concrete value. Resolving $\mathbf{Val}(c)\mathbf{R}\,v$ into a statement in theory $\mathbf{T}$ requires a symbolic hypothesis $\mathcal{H}$. This is because different symbolic hypotheses have different transition functions. Therefore, prior to performing an equivalence query using a concrete hypothesis, \ouralgorithm{} checks to make sure the symbolic hypothesis $\mathcal{H}$ satisfies all constraints in $\mathcal{B}$. If so, then an equivalence query can be made with a concrete hypothesis. If the hypothesis is correct, then \ouralgorithm{} terminates. Otherwise, $\mathcal{B}$ is expanded with the feedback received from the teacher. Next, we discuss empirical results for \ouralgorithm{}.

\section{Empirical Results}
In our investigation of \ouralgorithm{}, we are interested in the following questions: (1) How does \ouralgorithm{} scale? (2) How does learning different valuation functions affect runtime? (3) Does using strong feedback improve time to convergence? (4) Does explicitly expanding the table using counterexamples reduce time until termination? (5) What is the impact of encoding additional closed and consistency constraints as MaxSMT objectives? We describe the experimental setup used for answering these questions.

\begin{table*}[t!]
  \centering
  \caption{Valuation functions and required theories}
  \begin{tabular}{|c|c|c|}
    \hline
    \textbf{Valuation, Symbol} & \textbf{Output Alphabet} & \textbf{Theory Required} \\ \hline
    Non-Discounted Summation, $\mathbf{\Sigma}$ & Rationals & $\mathbf{T}_\mathbb{Q}$ (Theory of Rationals) \\ \hline
    Discounted Summation, $\mathbf{\gamma\Sigma}$& Rationals & $\mathbf{T}_\mathbb{Q}$ (Theory of Rationals) \\ \hline
    Product, $\mathbf{\Pi}$& Positive Rationals & $\mathbf{T}_\mathbb{Q}$ (Theory of Rationals) \\ \hline
    Classification, $\mathbf{\mathbb{N}}$& Finite Classes & Theory of Total Ordering \\ \hline
  \end{tabular}
  \label{tab:val-theory}
\end{table*}
\subsection{Experimental Setup}

We apply \ouralgorithm{} to infer a set of 25 ground truth quantitative automata of various input alphabet cardinalities and quantities of states, as a function of learning using theories associated the valuation functions listed in Table \ref{tab:val-theory}. The valuation functions include non-discounted summation, discounted summation, product, and classification valuations. We encode constraints using an exact rational representation, and use the Z3 SMT solver~\cite{demoura-z3-2008} to solve the MaxSMT objective.

For a baseline, we utilize \textsc{Remap}, which uses the classification valuation, but also includes greedy unification and strong feedback from the teacher, and therefore requires no backtracking. \ouralgorithm{} with counterexample expansion and strong feedback under classification valuation is most closely related to the baseline.

We compare \ouralgorithm{} and \textsc{Remap} using several variants of \ouralgorithm{}: (1) whether table expansion from counterexamples is enabled, (2) whether IDS is enabled, 
and (3) whether an additional a closedness and consistency objective (CC), described in the Appendix, is included in addition to the variable equivalence objective in Equation~\ref{eq:max-smt} (VE).

We run \ouralgorithm{} and \textsc{Remap} to learn a set of 25 ground truth automata. The teacher uses the appropriate valuation function when responding to preference queries and providing counterexample feedback. Equivalence queries are implemented by checking for automata equivalence via transition function and labels. We use the following symbols to denote valuation functions: $\mathbf{\Sigma}$ (summation), $\mathbf{\gamma\Sigma}$ (discounted summation), $\mathbf{\Pi}$ (product), and $\mathbf{\mathbb{N}}$ (classification). For each valuation function, we considered both strong (S) and weak (W) feedback, and we also considered the inclusion of an additional heuristic MaxSMT objective denoted by (CC) that ranks solutions by closedness, consistency, and number of states, with strongest preference for closedness and consistency. We denote these \ouralgorithm{} variations by S-VE (S feedback, CC disabled), W-VE (W feedback, CC disabled), S-CC-VE (S feedback, CC enabled), and W-CC-VE (W feedback, CC disabled). Each variation was run 100 times per ground truth.

We present results in Figure \ref{fig:scaling-results} for a subset of the 25 automata; the full set is available in the Appendix.

\begin{figure}[hp!]
    \centering
    \tiny{Number of Preference Queries}\\
    \includegraphics[width=\textwidth]{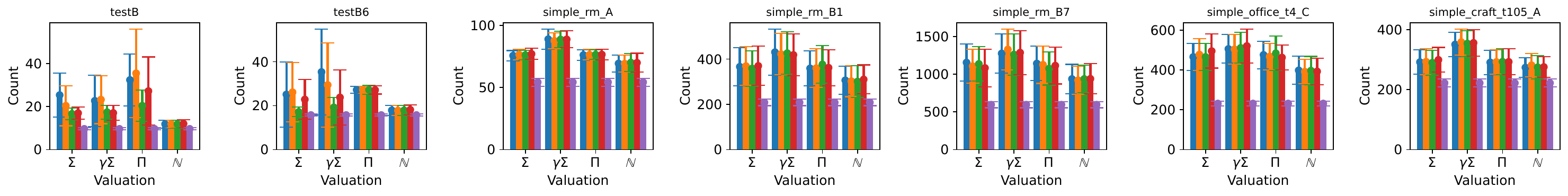}\\
    \tiny{Number of Equivalence Queries}\\
    \includegraphics[width=\textwidth]{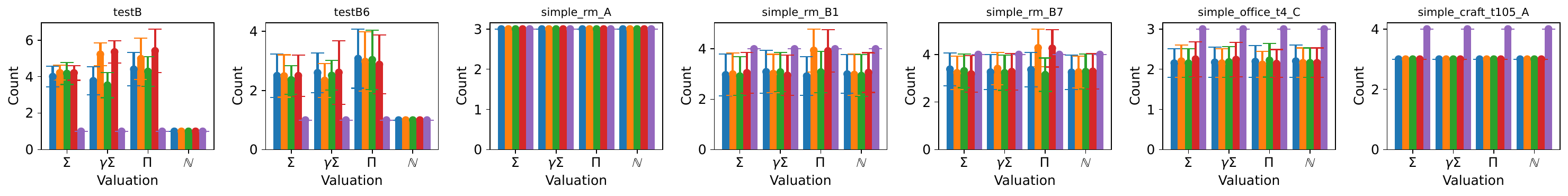}\\
    \tiny{Number of Inequalities Obtained}\\
    \includegraphics[width=\textwidth]{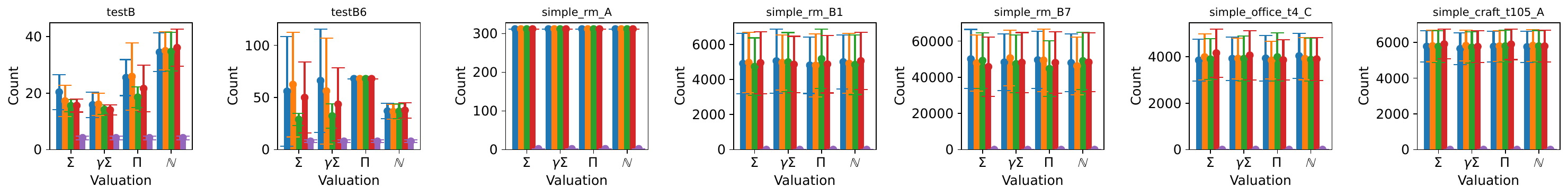}\\
    \tiny{Total Number of Variables}\\
    \includegraphics[width=\textwidth]{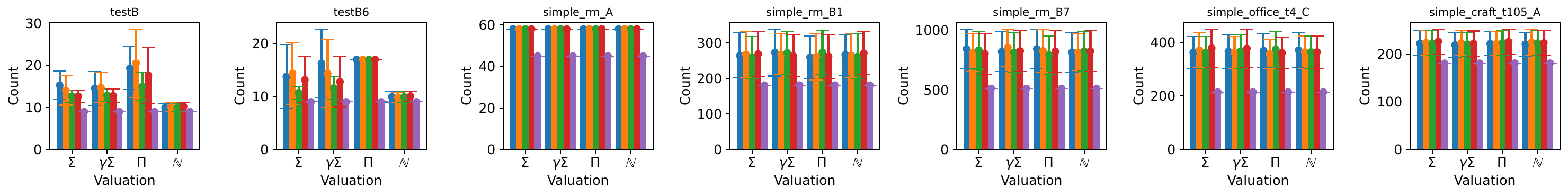}\\
    \tiny{Number of Unknown Variable Equivalences}\\
    \includegraphics[width=\textwidth]{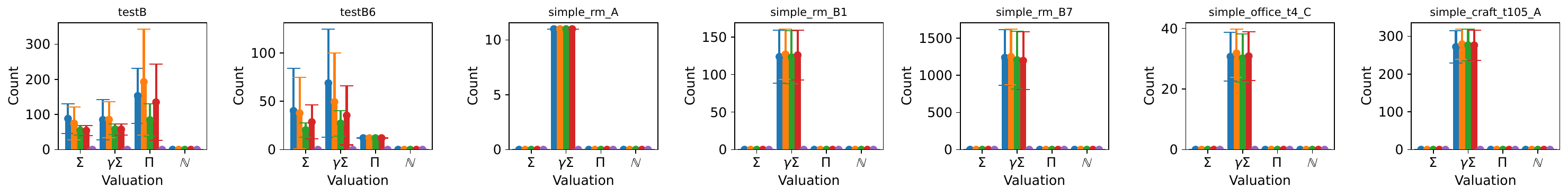}\\
    \tiny{Number of MaxSMT Objective Solves (overestimates Number of Conjectures Made)}\\
    \includegraphics[width=\textwidth]{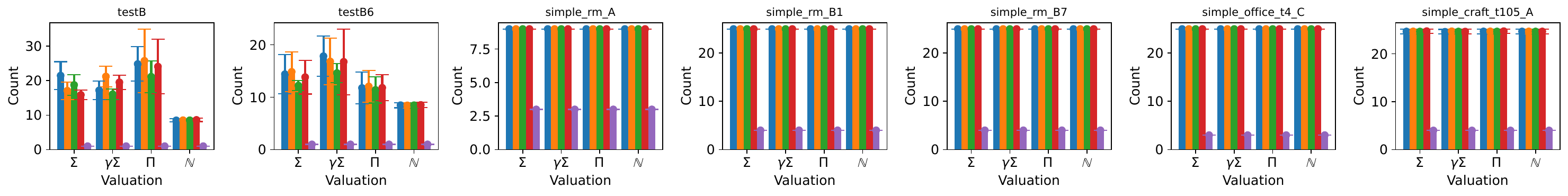}\\
    \tiny{Total Time Spent Solving MaxSMT Objectives in Milliseconds}\\
    \includegraphics[width=\textwidth]{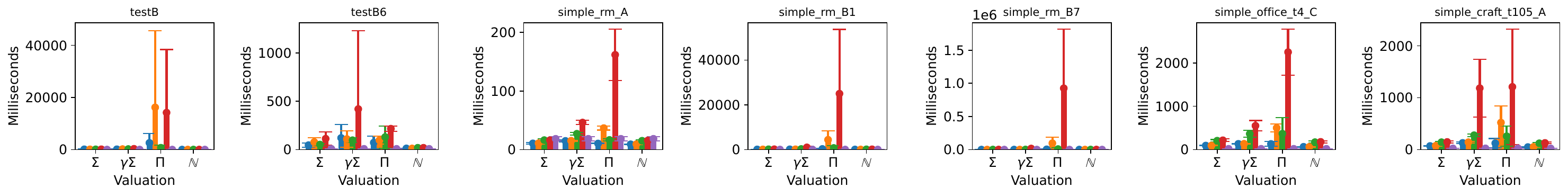}\\
    \tiny{Number of \textsc{SymbolicFill} Backtracking Calls}\\
    \includegraphics[width=\textwidth]{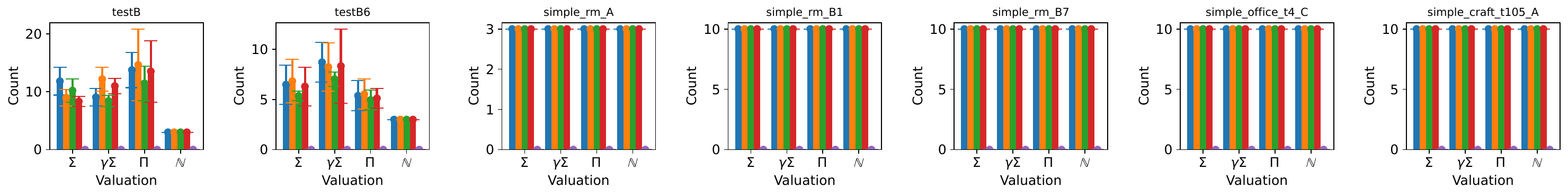}\\
    \includegraphics[width=0.5\textwidth]{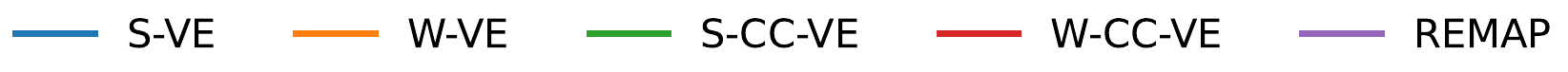}
    \caption{Results for 7 targets (columns) out of 25. Full results in Appendix. Each row is corresponds to a metric. Legend: S (strong feedback), W (weak feedback), VE (single MaxSMT objective Equation~\ref{eq:max-smt}), CC-VE (joint MaxSMT objectives closedness and consistency objective, and Equation~\ref{eq:max-smt}), \textsc{Remap} is the baseline. Each bar is the empirical mean of 100 trials, and error bars represent 1 standard deviation.}
    \label{fig:scaling-results}
\end{figure}
\subsection{Scaling}
We measure scaling by considering the difficulty of problems and the time taken to solve them. We measure the size of the MaxSMT objective (equivalent to the number of unknown variable equivalences), the number of variables created, number of preference and equivalence queries, number of inequalities gathered, and total time spent solving MaxSMT objectives. \ouralgorithm{} successfully learns all 25 targets. The results are shown in Figure \ref{fig:scaling-results}.

\textbf{Queries}. Across all environments, \ouralgorithm{} must make at least as many preference queries on average as \textsc{Remap}. The symbolic observation tables \ouralgorithm{} works with are generally larger than those of \textsc{Remap}, as measured by number of variables, due to the prefix-closed requirement on suffix set $E$ in \ouralgorithm{} which is needed to support recursive valuation function definitions. \textsc{Remap} does not require $E$ to be prefix-closed. The number of equivalence queries \ouralgorithm{} makes compared to \textsc{Remap} appears to be dependent on the target being learned, in terms of the minimum length of the conjecture--expansion chain required to reach the solution under IDS.

\textbf{Constraints and Solver Time}. The results shown in Figure \ref{fig:scaling-results} indicate that \ouralgorithm{} tends to gather significantly more inequalities than \textsc{Remap}, due to table expansions from incorrect variable equivalence conjectures. Searching an incorrect chain causes table expansions and induces preference queries, which gather inequalities. We also observe that \ouralgorithm{} can successfully navigate large numbers of unknown variable relations, as measured by MaxSMT objectives involving at least 1000 terms when learning certain targets under the $\mathbf{\gamma\Sigma}$ valuation. While the total number of conjectures made typically does not surpass $50$ when learning these targets, the total solver time can be quite expensive for certain targets, especially under weak feedback and when including the CC objective. In fact, the solving time appears to depend on the valuation function used, and therefore also depends on the corresponding solvers employed by Z3.

\subsection{Impact Of Valuation Function}
\textbf{Solver Time}. Our results show that \ouralgorithm{} has a more difficult time learning with $\mathbf{\gamma\Sigma}$ and $\mathbf{\Pi}$ valuation functions. We observe that for learning non-trivial targets, \ouralgorithm{} requires more solver time than \textsc{Remap} under all valuation functions $\mathbf{\Sigma}, \mathbf{\gamma\Sigma}, \mathbf{\Pi}, \text{ and }\mathbf{\mathbb{N}}$. Solving inequalities with $\mathbf{\Sigma}, \mathbf{\gamma\Sigma}, \text{ and }\mathbf{\mathbb{N}}$ remains under rational linear arithmetic, whereas solving under $\mathbf{\Pi}$ valuation results in polynomial inequalities that require a more expensive non-linear arithmetic solver. While a linearization using logarithms is possible, Z3 does not support transcendental functions.

\textbf{Conjectures and Objective Size}. We observe under $\mathbf{\mathbb{N}}$ valuation, \ouralgorithm{} does not encounter any unknown variable relations (measured by MaxSMT objective size), and that the solution is found with fewer MaxSMT solves than $\mathbf{\Sigma,\gamma\Sigma, \text{ or } \mathbf{\Pi}}$. The $\mathbf{\gamma\Sigma}$ and $\mathbf{\Pi}$ valuations generally result in the largest number of unknown variable relations.

\subsection{Impact Of Strong Feedback}
We hypothesize that strong feedback may be more informative than weak feedback. However, the results indicate use of strong feedback compared to weak feedback does not necessarily reduce the number of conjectures \ouralgorithm{} makes. An explanation is that the number of solutions to a strong feedback constraint is not necessarily smaller than the number of solutions to a weak feedback constraint. However, using strong feedback tends to use less solver time than weak feedback, which seems to indicate the constraint encoding is important for the solver. In contrast to feedback strength, including the closedness and consistency (CC) objective has greater impact on solver time and number of solves required.

\subsection{Effect of Closedness and Consistency Objective}
We observe that including the closedness and consistency MaxSMT objective incurs a trade-off. Solving two MaxSMT objectives takes more time than solving one objective, yet, including the CC objective slightly decreases number of (a) variables under consideration, (b) conjectures made, (c) unknown variable relations, and (d) inequalities gathered.

\subsection{Importance of Iterative Deepening and Counterexample Table Expansion}
We also consider ablations to \ouralgorithm{}. Our first ablation considers whether \ouralgorithm{} maintains a complete search without iterative deepening to limit conjecture--expansion chain length. When IDS is disabled, we observe learning under $\mathbf{\Pi}$ fails to terminate for certain targets. \ouralgorithm{} would search along an infinite length chain, under $\mathbf{\Pi}$, implying constant table expansions, even if the CC objective was enabled. However, under $\mathbf{\Sigma}, \mathbf{\gamma\Sigma}, \text{ and } \mathbf{\mathbb{N}}$ valuations, \ouralgorithm{} still successfully learns all 25 targets.

When disabling counterexample table expansion with IDS still enabled, we observe certain targets (e.g. such as \texttt{simple\_craft\_t105\_A}) become infeasible to learn due to memory usage. This is attributed to the fact that \ouralgorithm{} can detect when a table must be forcefully expanded, per Equation \ref{eqn:var-dim}, and expand the table. As shown in Appendix Figure~\ref{fig:ablation}, which measures the number of $(\mathcal{E,O})$ pairs considered, such expansion is not guided by counterexamples, and is therefore likely to expand the table in an unguided, breadth-first-search manner (adding the \emph{shortest} missing sequence) that is not optimal. This is evidenced by increased numbers of MaxSMT solves, and increased variance. However, while \ouralgorithm{} becomes much less efficient under these settings, it is still complete.

\section{Conclusion}
We present \ouralgorithm{}, an \lstar{}-based algorithm for actively learning deterministic quantitative automata from constraints obtained from preference queries and feedback. We consider how \ouralgorithm{} learns under summation, discounted summation, product, and classification valuation functions, and show how conjecturing variable equivalence under these valuation functions, guided by a MaxSMT objective, and limiting search depth by iterative deepening is essential for guaranteeing search completeness and minimalism in \ouralgorithm{}. Our empirical results include how \ouralgorithm{} scales, compares learning under different valuation functions, the impact of feedback strength, and the importance of iterative deepening and counterexample table expansion.

\bibliographystyle{splncs04nat}
\bibliography{sample}

\begin{thebibliography}{30}
\providecommand{\natexlab}[1]{#1}
\providecommand{\url}[1]{\texttt{#1}}
\providecommand{\urlprefix}{URL }
\expandafter\ifx\csname urlstyle\endcsname\relax
  \providecommand{\doi}[1]{doi:\discretionary{}{}{}#1}\else
  \providecommand{\doi}{doi:\discretionary{}{}{}\begingroup \urlstyle{rm}\Url}\fi

\bibitem[{Angluin(1987)}]{Angluin87}
Angluin, D.: Learning regular sets from queries and counterexamples. Inf. Comput. \textbf{75}(2), 87--106 (1987), \doi{10.1016/0890-5401(87)90052-6}, \urlprefix\url{https://doi.org/10.1016/0890-5401(87)90052-6}

\bibitem[{Argyros and D'antoni(2018)}]{Argyros2018TheLO}
Argyros, G., D'antoni, L.: The learnability of symbolic automata. In: International Conference on Computer Aided Verification (2018)

\bibitem[{Balle and Mohri(2015)}]{Balle2015LearningWA}
Balle, B., Mohri, M.: Learning weighted automata. In: Conference on Algebraic Informatics (2015)

\bibitem[{Bergadano and Varricchio(1994)}]{Bergadano1994LearningBO}
Bergadano, F., Varricchio, S.: Learning behaviors of automata from multiplicity and equivalence queries. SIAM J. Comput. \textbf{25}, 1268--1280 (1994)

\bibitem[{Biyik et~al.(2024)Biyik, Anari, and Sadigh}]{biyik2024batch}
Biyik, E., Anari, N., Sadigh, D.: Batch active learning of reward functions from human preferences. ACM Transactions on Human-Robot Interaction \textbf{13}(2), 1--27 (2024)

\bibitem[{Boker(2021)}]{qvsw}
Boker, U.: Quantitative vs. weighted automata. In: Reachability Problems: 15th International Conference, RP 2021, Liverpool, UK, October 25–27, 2021, Proceedings, p. 3–18, Springer-Verlag, Berlin, Heidelberg (2021), ISBN 978-3-030-89715-4, \doi{10.1007/978-3-030-89716-1\_1}, \urlprefix\url{https://doi.org/10.1007/978-3-030-89716-1\_1}

\bibitem[{Boker(2024)}]{discounted}
Boker, U.: Discounted-sum automata with real-valued discount factors. In: Proceedings of the 39th Annual ACM/IEEE Symposium on Logic in Computer Science, LICS '24, Association for Computing Machinery, New York, NY, USA (2024), ISBN 9798400706608, \doi{10.1145/3661814.3662090}, \urlprefix\url{https://doi.org/10.1145/3661814.3662090}

\bibitem[{Chalupa et~al.(2024)Chalupa, Henzinger, Mazzocchi, and Saraç}]{chalupa2024quak}
Chalupa, M., Henzinger, T.A., Mazzocchi, N., Saraç, N.E.: Quak: Quantitative automata kit (2024)

\bibitem[{Christiano et~al.(2017)Christiano, Leike, Brown, Martic, Legg, and Amodei}]{christiano2017deep}
Christiano, P.F., Leike, J., Brown, T., Martic, M., Legg, S., Amodei, D.: Deep reinforcement learning from human preferences. Advances in neural information processing systems \textbf{30} (2017)

\bibitem[{De~Gemmis et~al.(2009)De~Gemmis, Iaquinta, Lops, Musto, Narducci, Semeraro et~al.}]{de2009preference}
De~Gemmis, M., Iaquinta, L., Lops, P., Musto, C., Narducci, F., Semeraro, G., et~al.: Preference learning in recommender systems. Preference Learning \textbf{41}(41-55), 48 (2009)

\bibitem[{De~Moura and Bj\o{}rner(2008)}]{demoura-z3-2008}
De~Moura, L., Bj\o{}rner, N.: Z3: An efficient smt solver. In: Proceedings of the Theory and Practice of Software, 14th International Conference on Tools and Algorithms for the Construction and Analysis of Systems, p. 337–340, TACAS'08/ETAPS'08, Springer-Verlag, Berlin, Heidelberg (2008), ISBN 3540787992

\bibitem[{Dohmen et~al.(2022)Dohmen, Topper, Atia, Beckus, Trivedi, and Velasquez}]{dohmen-2022-icaps}
Dohmen, T., Topper, N., Atia, G.K., Beckus, A., Trivedi, A., Velasquez, A.: Inferring probabilistic reward machines from non-markovian reward signals for reinforcement learning. In: Kumar, A., Thi{\'{e}}baux, S., Varakantham, P., Yeoh, W. (eds.) Proceedings of the Thirty-Second International Conference on Automated Planning and Scheduling, {ICAPS} 2022, Singapore (virtual), June 13-24, 2022, pp. 574--582, {AAAI} Press (2022), \urlprefix\url{https://ojs.aaai.org/index.php/ICAPS/article/view/19844}

\bibitem[{Drews and D'antoni(2017)}]{symbolicAutomata}
Drews, S., D'antoni, L.: Learning symbolic automata. In: International Conference on Tools and Algorithms for Construction and Analysis of Systems (2017)

\bibitem[{Fisman and Saadon(2022)}]{fismanlattice}
Fisman, D., Saadon, S.: Learning and characterizing fully-ordered lattice automata. In: Bouajjani, A., Hol{\'i}k, L., Wu, Z. (eds.) Automated Technology for Verification and Analysis, pp. 266--282, Springer International Publishing, Cham (2022), ISBN 978-3-031-19992-9

\bibitem[{Gaon and Brafman(2020)}]{GaonB20_nonmarkovian}
Gaon, M., Brafman, R.I.: Reinforcement learning with non-markovian rewards. In: The Thirty-Fourth {AAAI} Conference on Artificial Intelligence, {AAAI} 2020, The Thirty-Second Innovative Applications of Artificial Intelligence Conference, {IAAI} 2020, The Tenth {AAAI} Symposium on Educational Advances in Artificial Intelligence, {EAAI} 2020, New York, NY, USA, February 7-12, 2020, pp. 3980--3987, {AAAI} Press (2020)

\bibitem[{Hsiung et~al.(2023)Hsiung, Biswas, and Chaudhuri}]{remap}
Hsiung, E., Biswas, J., Chaudhuri, S.: Learning reward machines through preference queries over sequences (2023), \urlprefix\url{https://arxiv.org/abs/2308.09301}

\bibitem[{Icarte et~al.(2018)Icarte, Klassen, Valenzano, and McIlraith}]{icarte2018}
Icarte, R.T., Klassen, T., Valenzano, R., McIlraith, S.: Using reward machines for high-level task specification and decomposition in reinforcement learning. In: Dy, J., Krause, A. (eds.) Proceedings of the 35th International Conference on Machine Learning, Proceedings of Machine Learning Research, vol.~80, pp. 2107--2116, PMLR (10--15 Jul 2018), \urlprefix\url{https://proceedings.mlr.press/v80/icarte18a.html}

\bibitem[{Maccarini et~al.(2022)Maccarini, Pura, Piga, Roveda, Mantovani, and Braghin}]{MACCARINI20227}
Maccarini, M., Pura, F., Piga, D., Roveda, L., Mantovani, L., Braghin, F.: Preference-based optimization of a human-robot collaborative controller. IFAC-PapersOnLine \textbf{55}(38), 7--12 (2022), ISSN 2405-8963, \doi{https://doi.org/10.1016/j.ifacol.2023.01.126}, \urlprefix\url{https://www.sciencedirect.com/science/article/pii/S2405896323001337}, 13th IFAC Symposium on Robot Control SYROCO 2022

\bibitem[{Maroto-Gomez et~al.(2022)Maroto-Gomez, Villarroya, Malfaz, Castro-Gonzalez, Castillo, and Salichs}]{maroto2022}
Maroto-Gomez, M., Villarroya, S.M., Malfaz, M., Castro-Gonzalez, A., Castillo, J.C., Salichs, M.A.: A preference learning system for the autonomous selection and personalization of entertainment activities during human-robot interaction. In: 2022 IEEE International Conference on Development and Learning (ICDL), pp. 343--348 (2022), \doi{10.1109/ICDL53763.2022.9962204}

\bibitem[{Narcomey et~al.(2024)Narcomey, Tsoi, Desai, and Vázquez}]{narcomey2024learninghumanpreferencesrobot}
Narcomey, A., Tsoi, N., Desai, R., Vázquez, M.: Learning human preferences over robot behavior as soft planning constraints (2024), \urlprefix\url{https://arxiv.org/abs/2403.19795}

\bibitem[{Sadigh et~al.(2017)Sadigh, Dragan, Sastry, and Seshia}]{Sadigh2017ActivePL}
Sadigh, D., Dragan, A.D., Sastry, S., Seshia, S.A.: Active preference-based learning of reward functions. In: Robotics: Science and Systems (2017)

\bibitem[{Shah et~al.(2023)Shah, Vazquez-Chanlatte, Junges, and Seshia}]{shah2023learning}
Shah, A., Vazquez-Chanlatte, M., Junges, S., Seshia, S.A.: Learning formal specifications from membership and preference queries (2023)

\bibitem[{Stiennon et~al.(2020)Stiennon, Ouyang, Wu, Ziegler, Lowe, Voss, Radford, Amodei, and Christiano}]{stiennon2020learning}
Stiennon, N., Ouyang, L., Wu, J., Ziegler, D., Lowe, R., Voss, C., Radford, A., Amodei, D., Christiano, P.F.: Learning to summarize with human feedback. Advances in Neural Information Processing Systems \textbf{33}, 3008--3021 (2020)

\bibitem[{Tappler et~al.(2019)Tappler, Aichernig, Bacci, Eichlseder, and Larsen}]{tappler2019based}
Tappler, M., Aichernig, B.K., Bacci, G., Eichlseder, M., Larsen, K.G.: {L*-based learning of Markov decision processes}. In: International Symposium on Formal Methods, pp. 651--669, Springer (2019)

\bibitem[{Tzeng(1992)}]{markovchain}
Tzeng, W.G.: Learning probabilistic automata and markov chains via queries. Mach. Learn. \textbf{8}(2), 151–166 (Mar 1992), ISSN 0885-6125, \doi{10.1023/A:1022616503659}, \urlprefix\url{https://doi.org/10.1023/A:1022616503659}

\bibitem[{Weiss et~al.(2019)Weiss, Goldberg, and Yahav}]{Weiss2019}
Weiss, G., Goldberg, Y., Yahav, E.: Learning deterministic weighted automata with queries and counterexamples. In: Wallach, H., Larochelle, H., Beygelzimer, A., d\textquotesingle Alch\'{e}-Buc, F., Fox, E., Garnett, R. (eds.) Advances in Neural Information Processing Systems, vol.~32, Curran Associates, Inc. (2019), \urlprefix\url{https://proceedings.neurips.cc/paper\_files/paper/2019/file/d3f93e7766e8e1b7ef66dfdd9a8be93b-Paper.pdf}

\bibitem[{Wirth et~al.(2017)Wirth, Akrour, Neumann, and F{\"u}rnkranz}]{wirth2017survey}
Wirth, C., Akrour, R., Neumann, G., F{\"u}rnkranz, J.: A survey of preference-based reinforcement learning methods. Journal of Machine Learning Research \textbf{18}(136), 1--46 (2017)

\bibitem[{Wirth et~al.(2016)Wirth, F{\"u}rnkranz, and Neumann}]{wirth2016model}
Wirth, C., F{\"u}rnkranz, J., Neumann, G.: Model-free preference-based reinforcement learning. In: Proceedings of the AAAI conference on artificial intelligence, vol.~30 (2016)

\bibitem[{Xu et~al.(2021)Xu, Wu, Ojha, Neider, and Topcu}]{xu_lstar}
Xu, Z., Wu, B., Ojha, A., Neider, D., Topcu, U.: Active finite reward automaton inference and reinforcement learning using queries and counterexamples. In: Machine Learning and Knowledge Extraction: 5th IFIP TC 5, TC 12, WG 8.4, WG 8.9, WG 12.9 International Cross-Domain Conference, CD-MAKE 2021, Virtual Event, August 17–20, 2021, Proceedings, p. 115–135, Springer-Verlag, Berlin, Heidelberg (2021), ISBN 978-3-030-84059-4, \doi{10.1007/978-3-030-84060-0\_8}, \urlprefix\url{https://doi.org/10.1007/978-3-030-84060-0\_8}

\bibitem[{Xue et~al.(2023)Xue, Cai, Xue, Sun, Liu, Zheng, Jiang, Gai, and An}]{xueprefrec2023kdd}
Xue, W., Cai, Q., Xue, Z., Sun, S., Liu, S., Zheng, D., Jiang, P., Gai, K., An, B.: Prefrec: Recommender systems with human preferences for reinforcing long-term user engagement. In: Proceedings of the 29th ACM SIGKDD Conference on Knowledge Discovery and Data Mining, p. 2874–2884, KDD '23, Association for Computing Machinery, New York, NY, USA (2023), ISBN 9798400701030, \doi{10.1145/3580305.3599473}, \urlprefix\url{https://doi.org/10.1145/3580305.3599473}

\end{thebibliography}

\newpage
\appendix
\section{Appendix}
\subsection{Closed and Consistency MaxSMT Objective}
We describe the closed and consistency objective (CC) here. In summary, the purpose of this objective is to encode a greedy preference into the MaxSMT objective for preferring variable equivalences which would make the current table closed and consistent after applying unification, and which would minimize the number of states in the table.

The objective function can be described as explicitly ranking all possible tables resulting from a single step of unification. There are 3 regimes of preferences:
\begin{enumerate}
    \item Closed and consistent tables are ranked highest.
    \item Closed, but not consistent table are ranked second.
    \item Tables which are not closed are ranked the lowest.
\end{enumerate}
Within the 3 regimes, tables are ranked by how many states they contain. Tables with fewer states and ranked higher than tables with more states.

The components of objective are constructed from 3 parts: (a) an expression which counts the number of states in the table, (b) an expression for a closed table, and (c) an expression for a consistent table.

Let the current table have $|S|=U$ upper rows and $|S\cdot\Sigma^I|=L$ lower rows. All expressions are developed based on row equivalence.

\subsubsection{State Minimization Expression} Let $C_n(S)$ be the set of tuples representing subsets of $\orows{S}$ with size $n$. There are $\binom{|S|}{n}$ such tuples in $C_n(S)$ since permutations are excluded. Then let
\begin{align}
\mathbf{M}(S,\Sigma^I) = \sum_{(u_1,u_2)\in C_2(S)} \mathbf{1}_{[u_1=u_2]} + \sum_{(l_1,l_2)\in C_2(S\cdot\Sigma^I)}\mathbf{1}_{[l_1=l_2]}+\sum_{(u,l)\in\orows{S}\times\orows{S\cdot\Sigma^I}}\mathbf{1}_{[u=l]}
\end{align}The expression on the right contains \begin{align}
    \displaystyle N=\frac{U(U-1)+L(L-1)}{2}+UL
\end{align} terms.

\subsubsection{Closedness Expression} For a table to be closed, we require that each row in $\orows{S\cdot\Sigma^I}$ be in $\orows{S}$. We express this requirement through explicit enumeration:
\begin{align}
\mathbf{C}_1(S,\Sigma^I)=\bigwedge_{l\in\orows{S\cdot\Sigma^I}}\left(\bigvee_{u\in\orows{S}} l=u\right)
\end{align}

\subsubsection{Consistency Expression} For a table to be consistent, we encode the definition of a consistent table using implications. Let $\mathbf{r}(\cdot)=\orow{\cdot}$ for succinctness. Then let $\mathbf{C}_2(S,\Sigma^I)=$
\begin{align}
    \bigwedge_{(\mathbf{r}(s_1),\mathbf{r}(s_2))\in C_2(S)}\left((\mathbf{r}(s_1)\equiv\mathbf{r}(s_2))\implies \bigwedge_{(\mathbf{r}(s_1\cdot\sigma),\mathbf{r}(s_2\cdot\sigma))\in C_2(S\cdot\Sigma^I)}(\mathbf{r}(s_1\cdot\sigma)\equiv\mathbf{r}(s_2\cdot\sigma))\right)
\end{align}
\subsubsection{Full Objective} We can now express the fill objective in terms of the above components:
\begin{align}
    \Lambda' = \argmax_{(v_1,\dots,v_m)}\left(\mathbf{M}(S,\Sigma^I)+\mathbf{1}_{[\mathbf{C}_1(S,\Sigma^I)]}\cdot\left[(N+1)\mathbf{1}_{[\mathbf{C}_2(S,\Sigma^I)]} + (N+1)\right]\right)
\end{align}
This formulation splits the solutions into the three regimes described above. Note the bounds $0\leq\mathbf{M}(S,\Sigma^I)\leq N$ imply that each solution will be ranked according to the regimes as appropriate.

\subsubsection{Combining Objectives} When we combine objectives, we do so lexographically. First, we solve for $\Lambda'$. Then, given $\Lambda'$, we solve Equation \ref{eq:max-smt} for $\Lambda$.

\subsection{Counterexample Ablation Results}
\begin{figure}
    \centering\tiny{Number of MaxSMT Objective Solves}\\
    \includegraphics[width=0.19\textwidth]{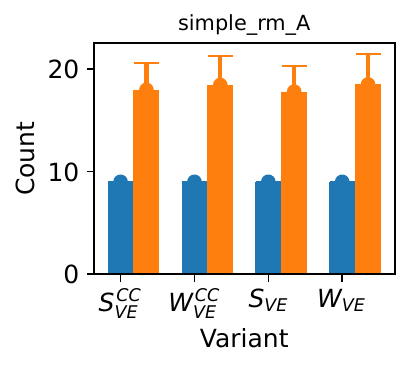}
    \includegraphics[width=0.19\textwidth]{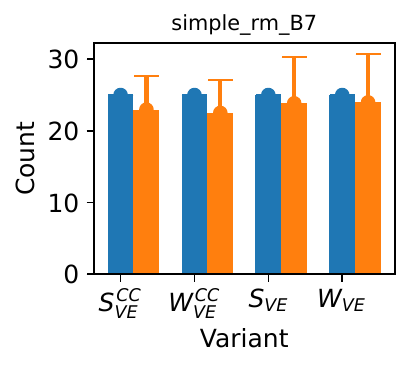}
    \includegraphics[width=0.19\textwidth]{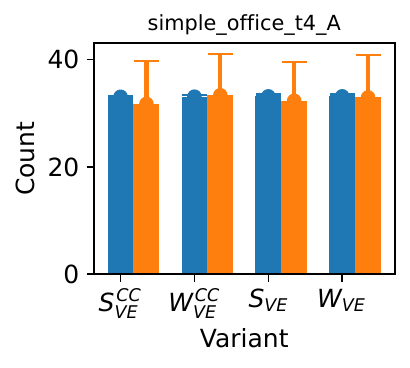}
    \includegraphics[width=0.19\textwidth]{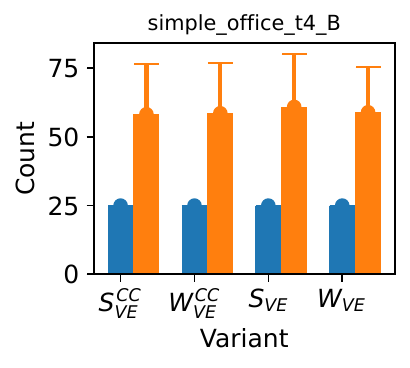}
    \includegraphics[width=0.19\textwidth]{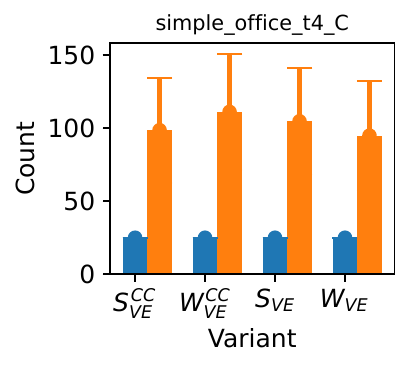}\\
    \includegraphics[width=0.4\textwidth]{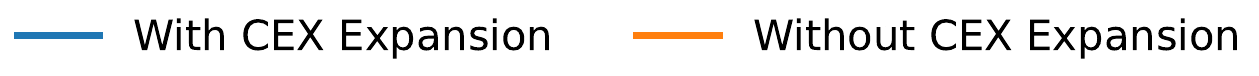}
    \caption{Counterexample Expansion Ablation, measuring number of MaxSMT objective solves, under non-discounted summation valuation, with IDS enabled. This measures the number of $(\mathcal{E,O})$ pairs considered.}
    \label{fig:ablation}
\end{figure}

\subsection{Full Results}
We provide the full range of plots for environments tested in.

\begin{figure}
    \centering
    \includegraphics[width=\textwidth]{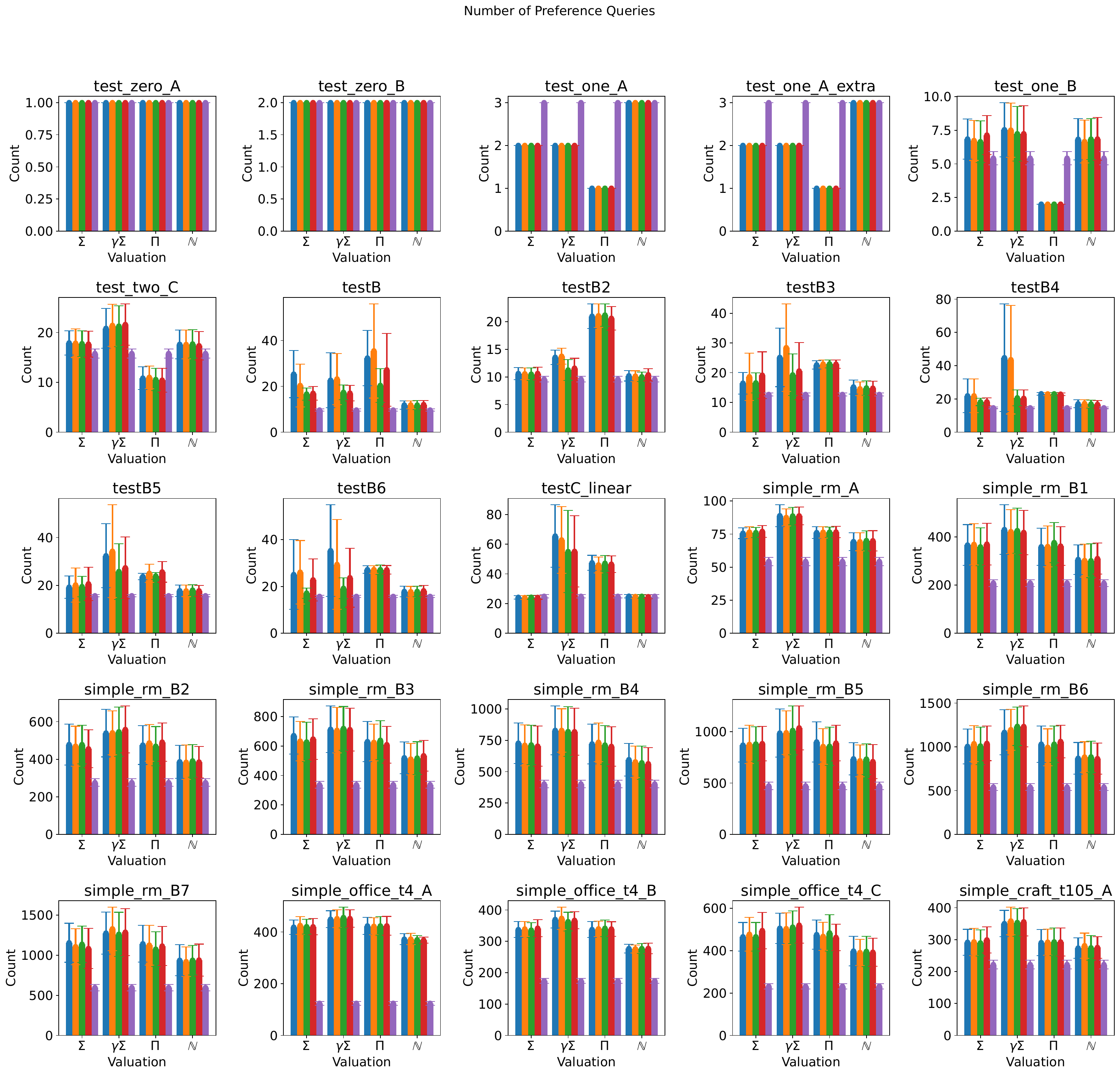}
    \includegraphics[width=0.5\textwidth]{figs/appendix/yes-cex-expansion/yes-b-cex-yes-ids/quintic_legend.pdf}
    \caption{Number of Preference Queries, under CEX Expansion and IDS}
    \label{fig:appendix-pref-q}
\end{figure}

\begin{figure}
    \centering
    \includegraphics[width=\textwidth]{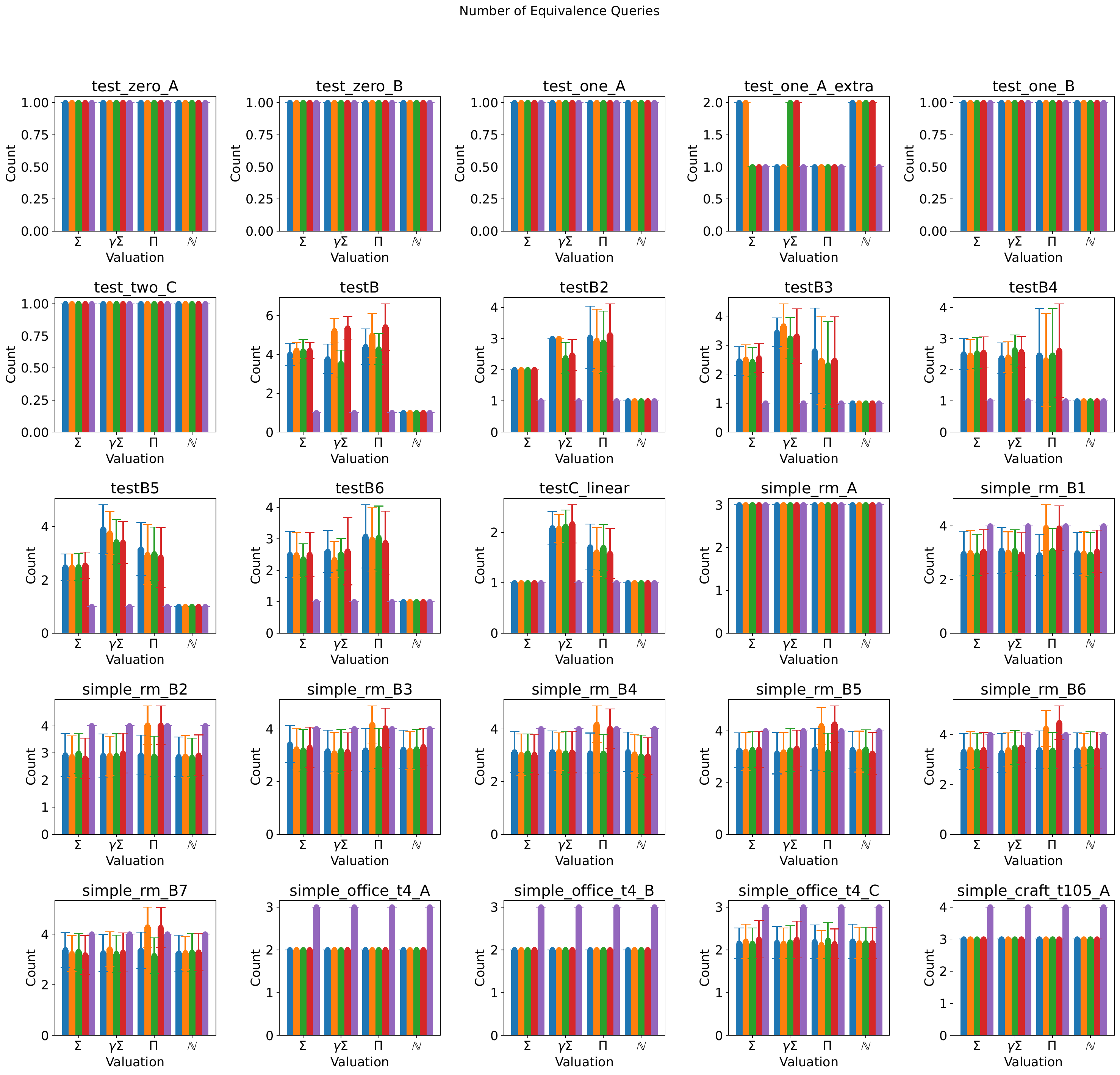}
    \includegraphics[width=0.5\textwidth]{figs/appendix/yes-cex-expansion/yes-b-cex-yes-ids/quintic_legend.pdf}
    \caption{Number of Equivalence Queries, under CEX Expansion and IDS}
    \label{fig:appendix-equiv-q}
\end{figure}

\begin{figure}
    \centering
    \includegraphics[width=\textwidth]{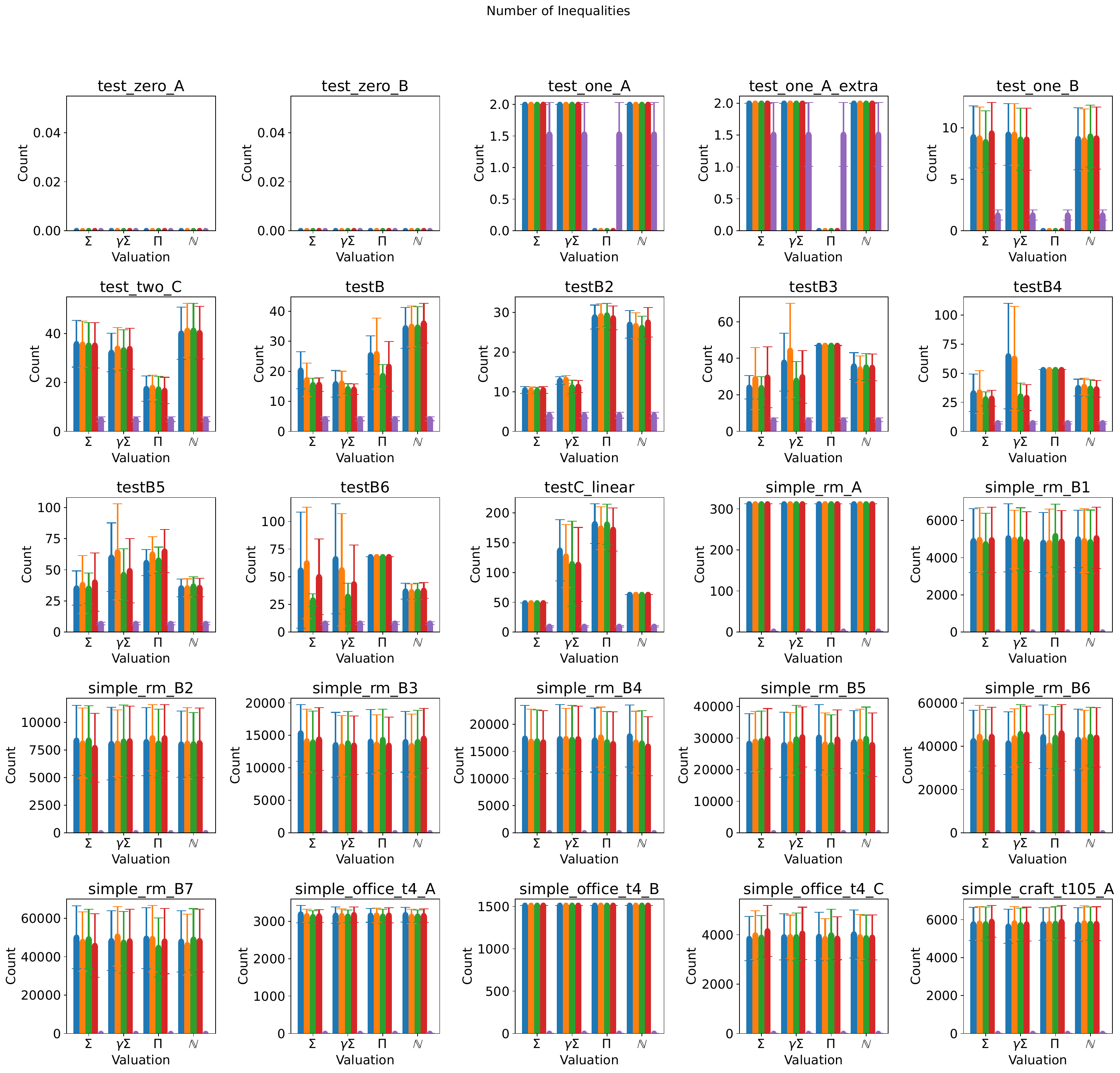}
    \includegraphics[width=0.5\textwidth]{figs/appendix/yes-cex-expansion/yes-b-cex-yes-ids/quintic_legend.pdf}
    \caption{Number of Inequalities Gathered, under CEX Expansion and IDS}
    \label{fig:appendix-num-ineq}
\end{figure}

\begin{figure}
    \centering
    \includegraphics[width=\textwidth]{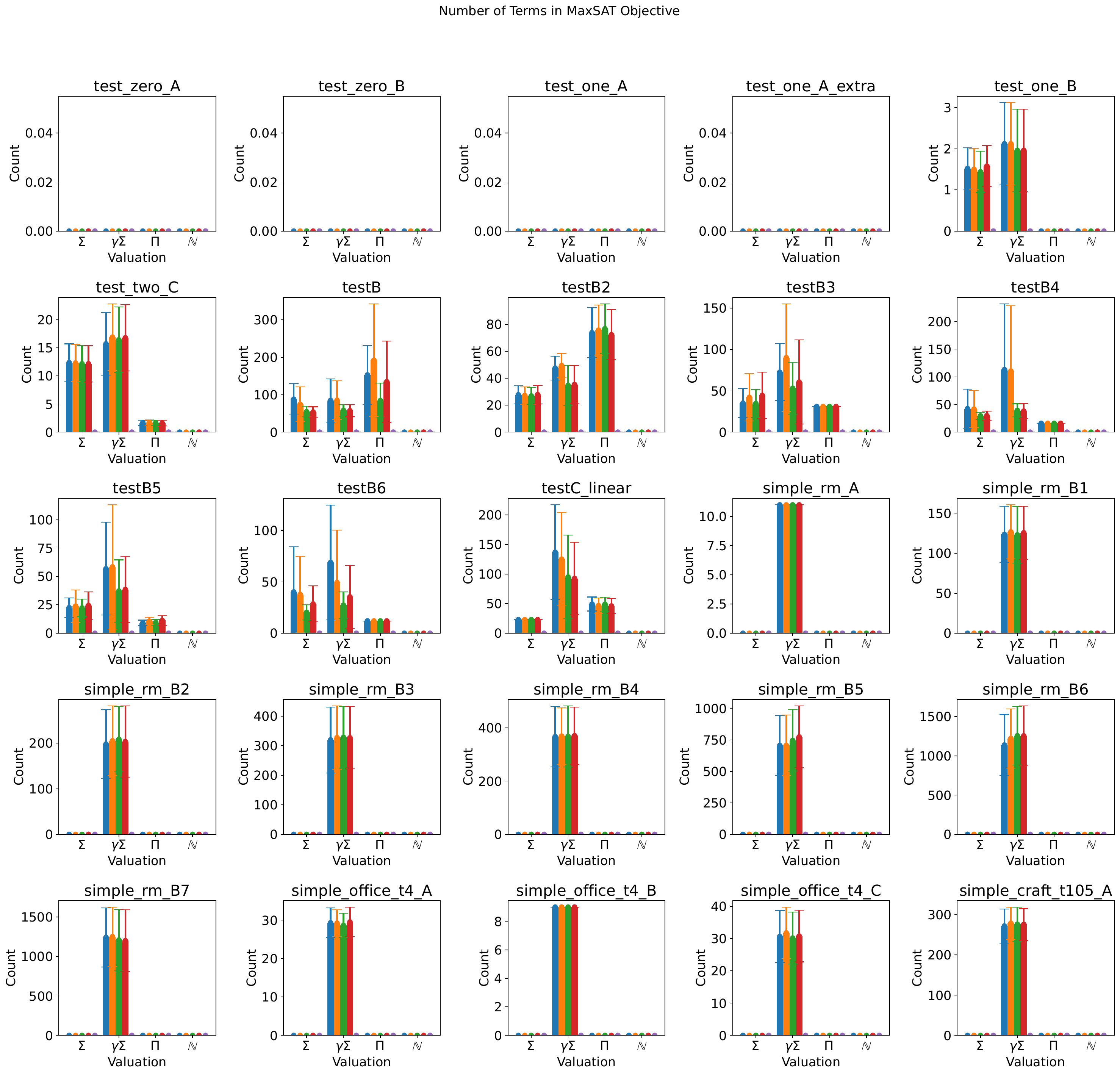}
    \includegraphics[width=0.5\textwidth]{figs/appendix/yes-cex-expansion/yes-b-cex-yes-ids/quintic_legend.pdf}
    \caption{Size of the MaxSMT Objective, under CEX Expansion and IDS}
    \label{fig:appendix-max-sat-objective-size}
\end{figure}

\begin{figure}
    \centering
    \includegraphics[width=\textwidth]{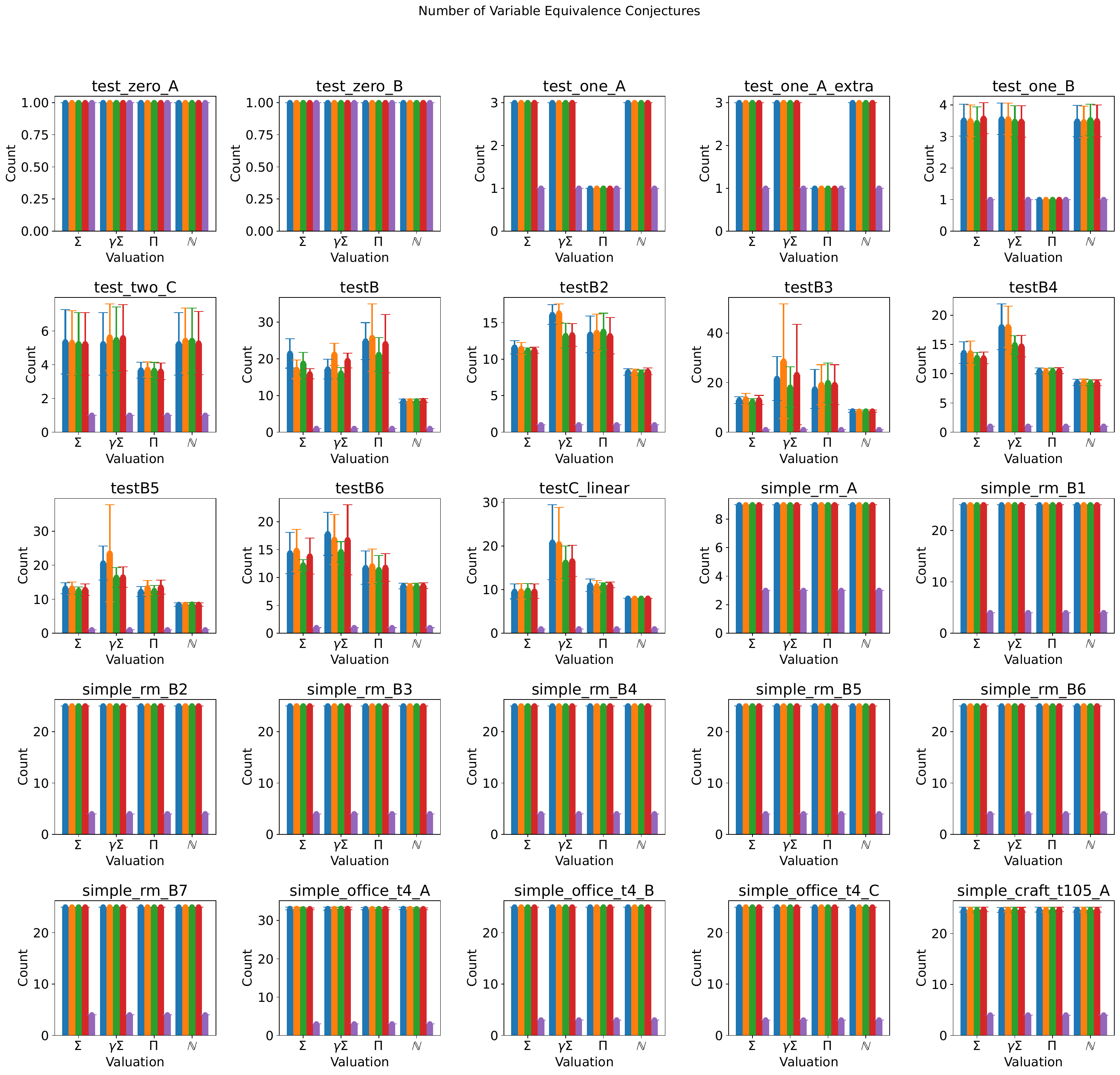}
    \includegraphics[width=0.5\textwidth]{figs/appendix/yes-cex-expansion/yes-b-cex-yes-ids/quintic_legend.pdf}
    \caption{Number of Variable Equivalence Conjectures, under CEX Expansion and IDS}
    \label{fig:appendix-num-ve-conjectures}
\end{figure}

\begin{figure}
    \centering
    \includegraphics[width=\textwidth]{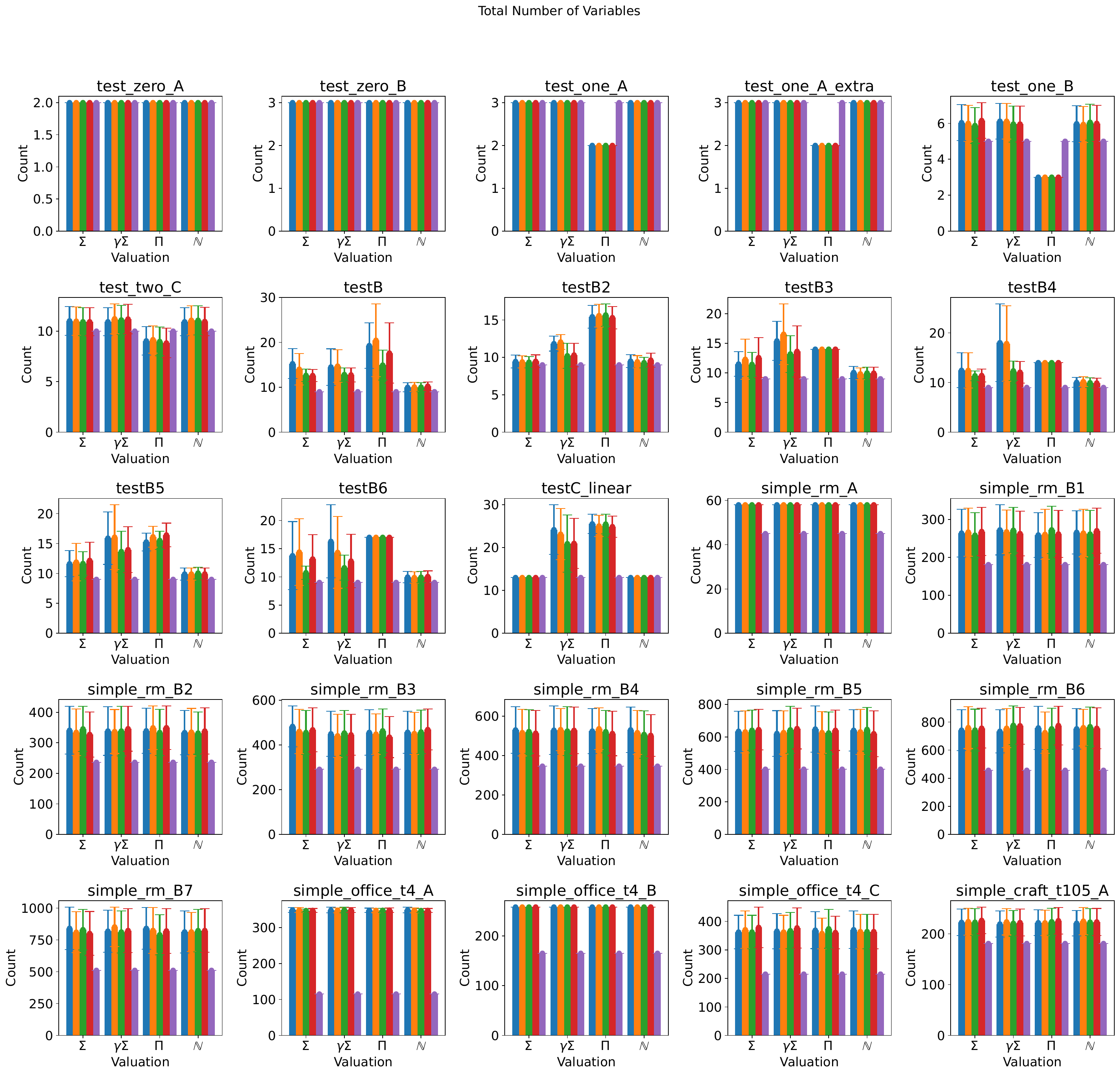}
    \includegraphics[width=0.5\textwidth]{figs/appendix/yes-cex-expansion/yes-b-cex-yes-ids/quintic_legend.pdf}
    \caption{Number of Variables, under CEX Expansion and IDS}
    \label{fig:appendix-num-variables-dimension}
\end{figure}

\begin{figure}
    \centering
    \includegraphics[width=\textwidth]{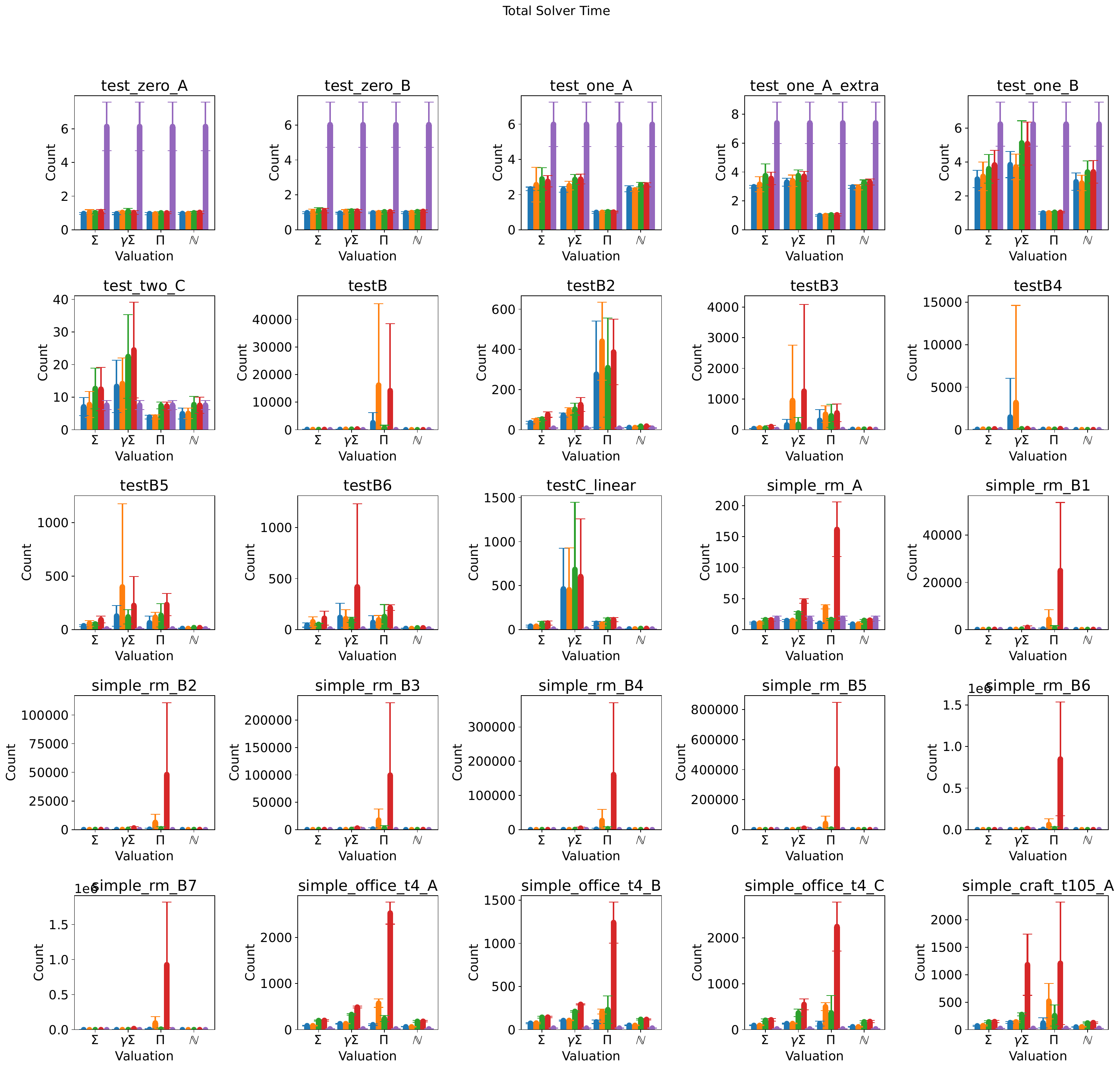}
    \includegraphics[width=0.5\textwidth]{figs/appendix/yes-cex-expansion/yes-b-cex-yes-ids/quintic_legend.pdf}
    \caption{Total Solver Time in Milliseconds, under CEX Expansion and IDS}
    \label{fig:appendix-solver-time}
\end{figure}

\begin{figure}
    \centering
    \includegraphics[width=\textwidth]{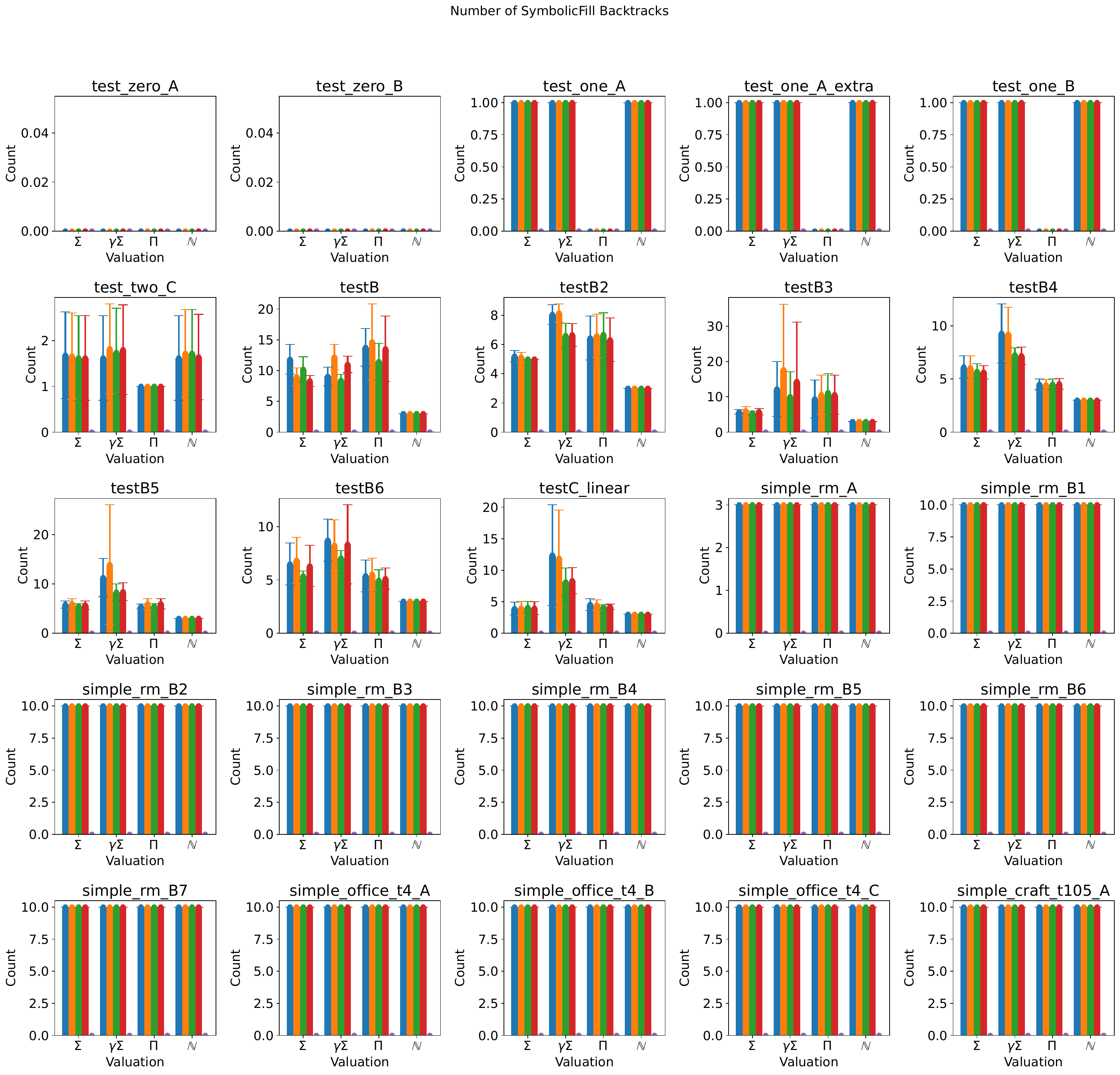}
    \includegraphics[width=0.5\textwidth]{figs/appendix/yes-cex-expansion/yes-b-cex-yes-ids/quintic_legend.pdf}
    \caption{Total Number of Backtracks, under CEX Expansion and IDS}
    \label{fig:appendix-backtracks}
\end{figure}
\end{document}